\documentclass[a4paper,11pt]{article}
\usepackage{amsmath, amsfonts, amssymb, amsthm, fullpage, graphicx}

\usepackage[T1]{fontenc}
\usepackage[utf8]{inputenc}
\usepackage{authblk}

\numberwithin{equation}{section}
\newtheorem{theorem}{Theorem}
\newtheorem{lemma}{Lemma}[section]
\newtheorem{corollary}[lemma]{Corollary}
\newtheorem{proposition}[lemma]{Proposition}
\newtheorem{remark}{Remark}

\newcommand{\RN}[1]{\textup{\uppercase\expandafter{\romannumeral#1}}}

\renewcommand\labelenumi{(\arabic{enumi})}
\renewcommand\theenumi\labelenumi

\allowdisplaybreaks

\begin{document}
	\title{Intertwining operator for $AG_2$ Calogero--Moser--Sutherland system}
	\author{Misha Feigin, Martin Vrabec}
\affil{School of Mathematics and Statistics, University of Glasgow, UK \thanks{Emails: Misha.Feigin@glasgow.ac.uk, 2197055V@student.gla.ac.uk}}
\date{}

	\maketitle

\begin{abstract}
	We consider generalised Calogero--Moser--Sutherland quantum Hamiltonian $H$ associated with a configuration of vectors $AG_2$ on the plane which is a union of $A_2$ and $G_2$ root systems. The Hamiltonian $H$ depends on one parameter. We find an intertwining operator between $H$ and the Calogero--Moser--Sutherland Hamiltonian for the root system $G_2$. This gives a quantum integral for $H$ of order 6 in an explicit form thus establishing integrability of $H$. 
\end{abstract}

\section{Introduction}
    The study of Calogero--Moser--Sutherland (CMS) integrable systems goes back to the works \cite{C'71} -- \cite{M'75}. Olshanetsky and Perelomov introduced generalised CMS systems related to root systems of Weyl groups  \cite{OP1}, which includes the non-reduced root system $BC_n$. The corresponding Hamiltonians are closely related to radial parts of Laplace-Beltrami operators on symmetric spaces \cite{OP2}, \cite{BPF}. In the case of root system $G_2$ the rational version of the corresponding CMS system was considered earlier by Wolfes \cite{W}. A uniform proof of integrability for all root sysyems via trigonometric version of Dunkl operators was given by Heckman in \cite{H}. Another more involved proof was provided earlier by Opdam in \cite{Op}. In the case of integer values of coupling parameters these CMS systems admit additional quantum integrals and they are algebraically integrable as it was established by Chalykh, Styrkas and Veselov in \cite{CSV} (see also \cite{CV}).
    
    It was found by Chalykh, Veselov and one of the authors in \cite{CFV'96}, \cite{CFV'98} that there are integrable generalisations of CMS type quantum systems which correspond to special configurations of vectors generalising root systems. Examples of such configurations include deformations of the root systems $A_n$ and $C_n$. These examples are related to symmetric superspaces, \cite{S'01}--\cite{SV2}, and to special representations of Cherednik algebras \cite{F'12}. The corresponding configurations of vectors have to satisfy so-called locus conditions \cite{CFV'99}. It is expected that there are very few such configurations, but they are not classified yet. We refer to \cite{Ch07} for a survey of results on locus configurations and integrability of rational, trigonometric and elliptic generalised CMS systems (see also \cite{CEO} and reference therein for the elliptic case). 


    The work \cite{FF} of Fairley and one of the authors deals with a class of   trigonometric locus configurations on the plane. In the process of classification of such configurations a new locus configuration $AG_2$ was found in \cite{FF} (see also \cite{F} where this configuration appears as well in different  but related context of WDVV equations). This configuration of vectors with multiplicities depends on one integer parmeter $m$. Being a locus configuration, it follows from the general results of Chalykh \cite{Ch07} that the corresponding generalised CMS operator $H$ has an intertwining relation with the Laplacian. This implies integrability and, moreover, algebraic integrability of the Hamiltonian $H$ with $m\in \mathbb Z$. The latter means existence of some additional quantum integrals;  see \cite{CV}, \cite{CSV} for more details on algebraic integrability including precise definition and examples of intertwining operators.


    
    Let us describe the configuration of vectors $AG_2$ in detail, following \cite{FF}. This is a non-reduced configuration consisting of the union of root systems $G_2$ and $A_2$. A positive half of this configuration is shown on the figure below. 
	
	\begin{center}
		\includegraphics[scale=0.6]{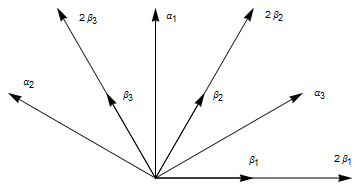}
	\end{center}
	
	The short vectors from the root system $G_2$ are denoted as $\pm \beta_1$, $\pm \beta_2$, $\pm \beta_3$, and the long vectors from the root system $G_2$ are denoted as $\pm \alpha_1$, $\pm \alpha_2$, $\pm \alpha_3$.
Let $\langle \cdot, \cdot \rangle$ be the standard Euclidean inner product on the plane. Then the ratio
$\frac{\langle \alpha_i, \alpha_i \rangle}{\langle \beta_i, \beta_i \rangle} = 3$ for any $i$.
 The vectors from the  additional root system $A_2$ are $\pm 2\beta_1$, $\pm 2\beta_2$, $\pm 2\beta_3$. Note that the numbering of $\alpha$'s is chosen in such a way that $ \langle \alpha_i, \beta_i \rangle = 0$ for all $i = 1, 2, 3$.
		
	Note that $$2 \beta_2 - \alpha_1 = \alpha_1 - 2 \beta_3 = \beta_1, \,\, 2 \beta_3 - \alpha_2 = \alpha_2 + 2 \beta_1 = \beta_2, \,\, 2 \beta_2 - \alpha_3 = \alpha_3 - 2 \beta_1 = \beta_3, $$ and $$\alpha_3 - \beta_2 = \beta_3 - \alpha_2 = \beta_2 - \beta_3 = \beta_1, \, \text{and } \alpha_1 - \beta_2 = \beta_3 \,.$$
	The configuration $AG_2$ is invariant under the $G_2$ Weyl group action and it belongs to a two-dimensional lattice spanned by $\beta_1$ and $\alpha_2$. However, it is not a root system since, for instance, $\beta_2 - \beta_1 = k \cdot 2 \beta_3 \,$, where $k = \frac12 \notin \mathbb{Z} \,$.
	
	Let $H_0$ be the CMS Hamiltonian for the root system $G_2$
	with multiplicities for the long and short roots $m$ and $3m$, respectively.
	And let $H$ be the Hamiltonian of the generalised CMS system under study, associated to the above collection of vectors $AG_2$ where
	$\alpha_i $, $\beta_i$ and $2 \beta_i$ have multiplicities $m$, $3m$ and $1$, respectively, where we use conventions for non-reduced systems coming from theory of symmetric spaces (see e.g. \cite{SV1}). 
    More precisely,
		\begin{equation*}
			H_0 = -\Delta + \sum_{i=1}^{3} \big( v_i(x) + u_i(x) \big) \,,
		\end{equation*}
		
		\begin{equation} \label{eq: AG2 Hamiltonian}
			H = -\Delta + \sum_{i=1}^{3} \big( v_i(x) + \widetilde{u}_i(x)  \big) \,,
		\end{equation}
    where $\Delta = \frac{\partial^2}{\partial x_1^2} + \frac{\partial^2}{\partial x_2^2}$ is the Laplacian, and for $x = (x_1, x_2) \in \mathbb{C}^2$ we have
    	\begin{alignat}{3}
   			&v_i(x) &&= \frac{m (m + 1) \langle \alpha_i, \alpha_i \rangle }{\sinh^2\langle \alpha_i, x \rangle}, \nonumber\\
   			&u_i(x) &&= \frac{3 m (3m + 1)\langle \beta_i, \beta_i \rangle }{\sinh^2\langle \beta_i, x \rangle} 
\label{eq: u's}\,,
    	\end{alignat}
    and 
    \begin{align*}
   		\quad \widetilde{u}_i(x)  & = \frac{9 m (m + 1) \langle \beta_i, \beta_i \rangle }{\sinh^2\langle \beta_i, x \rangle} + \frac{2 \langle 2\beta_i, 2\beta_i \rangle}{\sinh^2\langle 2\beta_i, x \rangle} \nonumber \\
   		& = \frac{(3 m + 1) (3m + 2) \langle \beta_i, \beta_i \rangle }{\sinh^2\langle \beta_i, x \rangle} - \frac{2 \langle \beta_i, \beta_i \rangle}{\cosh^2\langle \beta_i, x \rangle}. \nonumber 
    \end{align*}
    	
    In addition, we introduce the following notation for the difference $\widetilde{u}_i(x) - u_i(x)$:
   		\begin{equation} \label{eq: uhat's}
    		\widehat{u}_i(x) := \widetilde{u}_i(x) - u_i(x) =  \frac{2 (3 m + 1) \langle \beta_i, \beta_i \rangle}{\sinh^2\langle \beta_i, x \rangle} - \frac{2 \langle \beta_i, \beta_i \rangle}{\cosh^2\langle \beta_i, x \rangle}. 
     	\end{equation} 

     Let $\partial_i$ denote the partial derivative $\frac{\partial}{\partial x_i}$.
     For any vector (or a vector field) $\gamma = \big(\gamma^{(1)}, \gamma^{(2)} \big) \in \mathbb{C}^2$, we will write $\partial_\gamma$ for the directional derivative operator $\gamma^{(1)} \partial_1 + \gamma^{(2)} \partial_2$.
     In particular, if $\phi$ is a scalar field on the plane and $\nabla(\phi) = \big(\partial_1(\phi), \partial_2(\phi) \big)$ is its gradient, then by $\partial_{\nabla(\phi)}$ we will mean $\partial_1(\phi) \partial_1 + \partial_2(\phi) \partial_2 $.

In this paper we establish an intertwining relation between the Hamiltonian $H$ and the integrable Hamiltonian $H_0$ of the CMS system associated with the root system $G_2$. This relation is valid for any value of the parameter $m$  which is allowed to be non-integer. This leads to integrability of $H$ for any $m$ thus generalizing integrability for integer $m$ known from \cite{FF}, \cite{Ch07}. We also find the intertwining operator $\mathcal D$ of order 3 in an explicit form. This, in turn, gives quantum integral of $H$ of order 6. We note that direct application of results of \cite{Ch07} in the case of integer $m$ leads to a higher order  intertwiner and a higher order integral of $H$. The degree 6 for the integral of $H$ is expected to be minimal possible. Indeed, it follows from \cite{KT} that for generic $m$ an independent integral for the rational version of $H$ with constant highest term has to be of degree at least 6 since such highest term should be $G_2$-invariant.


The intertwining operator $\mathcal{D}$ has the form
     	\begin{equation} \label{eq: D}
     	     \mathcal{D} = \partial_{\beta_1}\partial_{\beta_2}\partial_{\beta_3} + \sum_{\sigma \in A_3} f_{_{\sigma(1)}} \partial_{\beta_{\sigma(2)}} \partial_{\beta_{\sigma(3)}} + \sum_{i=1}^{3} g_i \partial_{\beta_i} + h  \,,
     	\end{equation}
     where $A_3 = \{id, \, (1, 2, 3), \,(1, 3, 2)\} $ is the alternating group on $3$ elements, 
     and $f_i$, $g_i$ ($i = 1,2,3$) and $h$ are some functions which we specify explicitly. We will use the notation $\sum_{\sigma}$ throughout the paper as a shorthand for the cyclic sum $\sum_{\sigma \in A_3}$.
     To be more precise, we will prove the following main theorem.
     \begin{theorem} \label{thm: intertwining relation}
     	There exists a third-order differential operator $\mathcal{D}$ of the form \eqref{eq: D} such that
     	\begin{equation}\label{eq: intertwining relation}
     	    H \mathcal{D} = \mathcal{D} H_0.
     	\end{equation}
     \end{theorem} 
     
     We obtain quantum integrability of $H$ and a quantum integral of motion as a direct corollary by making use of a general statement from \cite{Ch}.
     \begin{theorem} \label{mainthint}
     	Let $\mathcal{D}$ satisfy \eqref{eq: intertwining relation} and let $\mathcal{D}^{*}$ be the formal adjoint of $\mathcal{D}$. 
     	Let $I$ be any differential operator such that the commutator $[I, H_0] = 0$.
     	Then $\mathcal{D} I \mathcal{D}^{*}$ commutes with~$H$. In particular, $[\mathcal{D}\mathcal{D}^{*}, H] = 0$.
     \end{theorem} 
 
 	Indeed, formal conjugation of the relation \eqref{eq: intertwining relation} gives $\mathcal{D}^{*} H = H_0 \mathcal{D}^{*}$. Hence
 	$$H \mathcal{D} I \mathcal{D}^{*} = \mathcal{D} H_0 I \mathcal{D}^{*}= \mathcal{D} I H_0 \mathcal{D}^{*} = \mathcal{D} I \mathcal{D}^{*} H.$$
 	Note that for integer $m$ the operator $H_0$ is algebraically integrable as the commutative ring of quantum integrals is larger than the ring of $G_2$--invariants \cite{CV}, \cite{CSV}. Therefore this gives a way to see algebraic integrability of the operator $H$ for integer $m$ (see also \cite{Ch07}). We also note that in the rational limit the operator $\mathcal{D} \mathcal{D}^{*}$ reduces to a quantum integral for the rational CMS system associated with the root system $G_2$ with multiplicities $m$ and $3m+1$ for the long and short roots, respectively. 

The structure of the paper is as follows. We collect some preliminary trigonometric identities associated with vectors from the configuration $AG_2$ in Section \ref{prelim}. We introduce all the coefficients of the intertwining operator \eqref{eq: D} in Section \ref{section: D}, where we also establish some preliminary results on these coefficients. In Section \ref{section: intertwining relation} we prove the main Theorem \ref{thm: intertwining relation} on the intertwining relation. We present results on the rational limit in Section \ref{rational_limit}. We outline some future directions in Section~\ref{concluderem}.
     
\vspace{5mm}

{\bf Acknowledgments.} M.F. is  very grateful to Oleg Chalykh for useful discussions and comments. We are grateful to the London Mathematical Society for the support through Undergraduate Research Bursary scheme which enabled us to carry out the main part of the work in summer 2018. M.V. also acknowledges matched funding from the School of Mathematics and Statistics, University of Glasgow.

\section{Preliminary trigonometric identities}
\label{prelim}

    In this section we collect some trigonometric identities involving vectors from the configuration $AG_2$. We will use these identities later in the proof of the intertwining relation \eqref{eq: intertwining relation} in Section~\ref{section: intertwining relation}. 

One can choose a coordinate system where the vectors take the form 
	$\beta_1 = \omega \big(\sqrt{2}, 0 \big)$, $\beta_2 = \omega \big(\frac{\sqrt{2}}{2}, \frac{\sqrt{6}}{2} \big)$, $\beta_3 = \omega\big({-} \frac{\sqrt{2}}{2}, \frac{\sqrt{6}}{2} \big)$,
	$\alpha_1 = \omega\big(0, \sqrt{6} \big)$, $\alpha_2 = \omega\big({-}\frac{3\sqrt{2}}{2}, \frac{\sqrt{6}}{2} \big)$ and $\alpha_3 = \omega \big(\frac{3\sqrt{2}}{2}, \frac{\sqrt{6}}{2} \big)$ for some non-zero $\omega \in {\mathbb C}$. We will use inner products between the vectors but not the specific coordinates  of the vectors.

\begin{remark} \label{rem: variants of identities}
	In most cases we state only one particular form of each identity, but other variants can be obtained by rotating or scaling the vectors.
	More precisely, the relevant transformations will be the replacement of $\beta_i$ with $2 \beta_i \,$, and the two rotations by $\frac{\pi}{3}$, clockwise and anti-clockwise. These rotations can alternatively be defined by the following replacement rules: 
	$ \beta_1 \rightarrow {-}\beta_3 \,, \, \beta_2 \rightarrow \beta_1 \,, \, \beta_3 \rightarrow \beta_2 \,, \, \alpha_1 \rightarrow \alpha_3 \,, \, 
	\alpha_2 \rightarrow \alpha_1 \,, \, \alpha_3 \rightarrow {-}\alpha_2 \,,$ and
	$ \beta_1 \rightarrow \beta_2 \,, \, \beta_2 \rightarrow \beta_3 \,, \, \beta_3 \rightarrow {-}\beta_1 \,, \, \alpha_1 \rightarrow \alpha_2 \,, \, 
	\alpha_2 \rightarrow {-}\alpha_3 \,, \, \alpha_3 \rightarrow \alpha_1 \,,$ respectively. 
\end{remark}

The vectors $\alpha_i$ and $\beta_i$ in the configuration $AG_2$ satisfy the following trigonometric identities, where we omit the argument $x = \big( x_1, x_2 \big)$. Thus we write $\coth \beta_i$ for $\coth \langle \beta_i, x \rangle$, etc.
\begin{lemma} \label{lemma: sum of products of coths}
	We have 
	\begin{equation}\label{eq: sum of products of coths}
	    \sum_{1 \leq j < k \leq 3} \langle \beta_j, \beta_k \rangle \coth \beta_j \coth \beta_k  = \omega^2 \,. 
	\end{equation}
\end{lemma}
\begin{proof}
	By a difference of cotangents formula and the fact that $\beta_2 - \beta_3 = \beta_1$, two terms in~the~sum~become 
	\begin{equation} \label{eq: first two terms}
		\omega^2 \coth \beta_1 \big( \coth \beta_2 - \coth \beta_3 \big) = - \frac{\omega^2\cosh \beta_1}{\sinh \beta_2 \sinh \beta_3} \,.
	\end{equation}
	By expanding $\cosh \beta_1$ in terms of $\beta_2$ and $\beta_3\,$, we can rearrange the right-hand side of \eqref{eq: first two terms} further as
	\begin{equation*}
	    - \frac{\omega^2\cosh \beta_2 \cosh \beta_3 - \omega^2 \sinh \beta_2 \sinh \beta_3}{\sinh \beta_2 \sinh \beta_3} = -\omega^2 \coth \beta_2 \coth \beta_3 + \omega^2 \,,
	\end{equation*}  
	as required.
\end{proof}

It will be convenient to use the following notation throughout the paper:
	\begin{align}
        X &= \omega^2( \sinh \beta_1 \sinh \beta_2 \sinh \beta_3 )^{-1} \label{eq: X} \,, \\
        Y &= \omega^2 (\sinh 2\beta_1 \sinh 2\beta_2 \sinh 2\beta_3)^{-1} \label{eq: Y} \,.
    \end{align}
Here are some identities involving these functions.

\begin{lemma}  \label{lemma: alt expression for X}
	We have 
	\begin{equation}  \label{eq: alt expression for hIV}
	    	     \sum_{\sigma} \frac{\langle \beta_{\sigma(2)}, \beta_{\sigma(3)} \rangle \coth \beta_{\sigma(1)}}{\sinh^2 \beta_{\sigma(2)} \sinh^2 \beta_{\sigma(3)}}
	    	    = -2 X \,.
	\end{equation}
\end{lemma}
\begin{proof}
	Multiplying equality \eqref{eq: sum of products of coths} by $2 \omega^{-2}$ and regrouping terms as in the proof of Lemma \ref{lemma: sum of products of coths}, we get 
	\begin{equation} \label{eq: alt expression for X rearranged}
		\begin{aligned}
			2 &= \coth \beta_1 \big( \coth \beta_2 - \coth \beta_3 \big) + \coth \beta_2 \big( \coth \beta_1 + \coth \beta_3 \big) + \coth \beta_3 \big( \coth \beta_2 - \coth \beta_1 \big) \\
			&= - \frac{\cosh \beta_1}{\sinh \beta_2 \sinh \beta_3} + \frac{\cosh \beta_2}{\sinh \beta_1 \sinh \beta_3} - \frac{\cosh \beta_3}{\sinh \beta_1 \sinh \beta_2} \,.
		\end{aligned} 
	\end{equation}
    The statement follows by dividing \eqref{eq: alt expression for X rearranged} by $\sinh \beta_1 \sinh \beta_2 \sinh \beta_3$.
\end{proof} 	

There is also the following version of Lemma \ref{lemma: alt expression for X} involving the function $\tanh$ rather than $\coth$.
\begin{lemma}  \label{lemma: alt expression for X and Y}
	We have
	\begin{equation} \label{eq: alt expression for hIV 2}
		-\frac12 \sum_{\sigma} \frac{\langle \beta_{\sigma(2)}, \beta_{\sigma(3)} \rangle \tanh \beta_{\sigma(1)}}{\sinh^2 \beta_{\sigma(2)} \sinh^2 \beta_{\sigma(3)}} = X + 4Y \,.
	\end{equation}
\end{lemma}

\begin{proof}
	Note the relation $\tanh z = \coth z - (\sinh z \cosh z)^{-1}$ valid for all $z \in \mathbb{C}$. Hence by Lemma \ref{lemma: alt expression for X} we~get 
	\begin{align*}
		&\sum_{\sigma} \frac{\langle \beta_{\sigma(2)}, \beta_{\sigma(3)} \rangle \tanh \beta_{\sigma(1)}}{\sinh^2 \beta_{\sigma(2)} \sinh^2 \beta_{\sigma(3)}} = 
		\sum_{\sigma} \frac{\langle \beta_{\sigma(2)}, \beta_{\sigma(3)} \rangle \coth \beta_{\sigma(1)}}{\sinh^2 \beta_{\sigma(2)} \sinh^2 \beta_{\sigma(3)}}  
	    - \sum_{\sigma} \frac{\langle \beta_{\sigma(2)}, \beta_{\sigma(3)} \rangle }{\sinh \beta_{\sigma(1)} \cosh \beta_{\sigma(1)} \sinh^2 \beta_{\sigma(2)} \sinh^2 \beta_{\sigma(3)}} \\
		&\qquad = - 2X - \frac{\sum_{1 \leq j < k \leq 3} \langle \beta_j, \beta_k \rangle \coth \beta_j \coth \beta_k }{\sinh \beta_1 \cosh \beta_1 \sinh \beta_2 \cosh \beta_2 \sinh \beta_3 \cosh \beta_3} \,.
	\end{align*}
	The result follows by applying Lemma \ref{lemma: sum of products of coths}.
\end{proof}

%

\begin{lemma} \label{lemma: derivative of X}
	We have 
	\begin{equation} \label{eq: derivative of X}
		 - X \sum_{i=1}^3 \coth \beta_i \partial_{\beta_i} = \sum_{\sigma} \frac{\langle \beta_{\sigma(2)}, \beta_{\sigma(3)} \rangle}{\sinh^2 \beta_{\sigma(2)} \sinh^2 \beta_{\sigma(3)}} \partial_{\beta_{\sigma(1)}} \,.
	\end{equation}
\end{lemma}

\begin{proof}
	Let us replace $\partial_{\beta_1} = \partial_{\beta_2} - \partial_{\beta_3} \,$ in \eqref{eq: derivative of X}. Then the coefficient of $\partial_{\beta_2}$ in the left-hand side equals
	\begin{equation*}
	    -X \big( \coth \beta_1 + \coth \beta_2 \big) = -\frac{\omega^2\sinh(\beta_1 + \beta_2)}{\sinh^2 \beta_1 \sinh^2 \beta_2 \sinh \beta_3} \,.
	\end{equation*} 
	Then on the right-hand side of equality \eqref{eq: derivative of X} the coefficient at $\partial_{\beta_2}$ is 
	\begin{equation*}
	    \frac{\omega^2}{\sinh^2 \beta_2 \sinh^2 \beta_3} -  \frac{\omega^2}{\sinh^2 \beta_1 \sinh^2 \beta_3} =
	    -\frac{\omega^2(\sinh^2 \beta_2 - \sinh^2 \beta_1)}{\sinh^2 \beta_1 \sinh^2 \beta_2 \sinh^2 \beta_3} =  -\frac{\omega^2 \sinh(\beta_1 + \beta_2)}{\sinh^2 \beta_1 \sinh^2 \beta_2 \sinh \beta_3}
	\end{equation*}
	as $\sinh^2 \beta_2 - \sinh^2 \beta_1 = \sinh(\beta_1 + \beta_2) \sinh(\beta_2 - \beta_1) \,$. Similarly, the coefficient at $\partial_{\beta_3}$ matches too.
\end{proof} 

As an immediate corollary of Lemma \ref{lemma: derivative of X} we get the following statement.
\begin{corollary} \label{cor: nabla of X}
    We have
	\begin{equation*}
		\partial_{\nabla(X)} = \sum_{\sigma} \frac{\langle \beta_{\sigma(2)}, \beta_{\sigma(3)} \rangle}{\sinh^2 \beta_{\sigma(2)} \sinh^2 \beta_{\sigma(3)}} \partial_{\beta_{\sigma(1)}} \,,
	\end{equation*} and
	\begin{equation*}
		\partial_{\nabla(Y)} = \sum_{\sigma} \frac{2 \langle \beta_{\sigma(2)}, \beta_{\sigma(3)} \rangle}{\sinh^2 2\beta_{\sigma(2)} \sinh^2 2\beta_{\sigma(3)}} \partial_{\beta_{\sigma(1)}} \,. 
	\end{equation*} \qed
\end{corollary} 		

The following lemma can be proven by a straightforward calculation with the help of Lemma \ref{lemma: sum of products of coths}.
\begin{lemma} \label{lemma: delta of X}
    The following two equalities hold:
	\begin{align}
	    \Delta (X) &= 
        4 \omega^2 \bigg(2 + \sum_{j=1}^3 \frac{1}{\sinh^2 \beta_j} \bigg)X \,, \label{eq: delta of inv prod of X} \\
	    \Delta (Y) &= 
       16 \omega^2\bigg(2 + \sum_{j=1}^3 \frac{1}{\sinh^2 2\beta_j} \bigg)Y \,. \label{eq: delta of inv prod of Y} 
	\end{align}
\end{lemma}

%

\begin{lemma} \label{lemma: first order identity}
	The following identities hold:
	\begin{align}
		- \bigg( \frac{1}{\cosh^2 \beta_2} + \frac{1}{\cosh^2 \beta_3} \bigg) \frac{1}{\sinh^2 \alpha_1} + 2 (\tanh \beta_2 + \tanh \beta_3) \frac{\coth \alpha_1}{\sinh^2 \alpha_1} &= \frac{1}{\cosh^2 \beta_2 \cosh^2 \beta_3}  \label{eq: first order identity 1} \,, \\
		\bigg( \frac{1}{\sinh^2 \beta_2} + \frac{1}{\sinh^2 \beta_3} \bigg) \frac{1}{\sinh^2 \alpha_1} + 2 (\coth \beta_2 + \coth \beta_3) \frac{\coth \alpha_1}{\sinh^2 \alpha_1} &= \frac{1}{\sinh^2 \beta_2 \sinh^2 \beta_3} \label{eq: first order identity 2} \,, \\
		- \bigg( \frac{1}{\cosh^2 \beta_2} + \frac{1}{\cosh^2 \beta_3} \bigg) \frac{1}{\sinh^2 \beta_1} + 2 (\tanh \beta_2 - \tanh \beta_3) \frac{\coth \beta_1}{\sinh^2 \beta_1} &= \frac{1}{\cosh^2 \beta_2 \cosh^2 \beta_3} \label{eq: first order identity 3} \,, \\
		\bigg( \frac{1}{\sinh^2 \beta_2} + \frac{1}{\sinh^2 \beta_3} \bigg) \frac{1}{\sinh^2 \beta_1} + 2 (\coth \beta_2 - \coth \beta_3) \frac{\coth \beta_1}{\sinh^2 \beta_1} &= \frac{1}{\sinh^2 \beta_2 \sinh^2 \beta_3} \label{eq: first order identity 4} \,.
	\end{align}
\end{lemma}

\begin{proof}
    Since $\beta_2 + \beta_3 = \alpha_1$, we have
	\[
		\tanh \beta_2 + \tanh \beta_3 = \tanh \alpha_1 (1 + \tanh \beta_2 \tanh \beta_3) \,.
	\]
	Therefore the left-hand side of the relation \eqref{eq: first order identity 1} multiplied by $\sinh^2 \alpha_1$ takes the form
$$
		-\frac{1}{\cosh^2 \beta_2} - \frac{1}{\cosh^2 \beta_3} + 2(1 + \tanh \beta_2 \tanh \beta_3)  
		= (\tanh \beta_2 + \tanh \beta_3)^2 = \frac{\sinh^2(\beta_2 + \beta_3)}{\cosh^2 \beta_2 \cosh^2 \beta_3} \,,
$$
	which implies the relation \eqref{eq: first order identity 1}.
	The equalities \eqref{eq: first order identity 2} -- \eqref{eq: first order identity 4} can be proved by following a similar sequence of steps, using in addition that $\coth x + \coth y = \tanh(x+y)(1+\coth x \coth y )$ for $x, y \in \mathbb{C}$.
\end{proof}

Several other identities can be derived from Lemma \ref{lemma: first order identity}, which we put in Lemmas \ref{lemma: zeroth order identity 1} and
\ref{lemma: zeroth order identity 2} below.
\begin{lemma} \label{lemma: zeroth order identity 1}
	The following relation is satisfied:
	\begin{equation} \label{eq: zeroth order identity 1}
        \bigg( \frac{\coth \beta_3}{\sinh^2 \beta_3} - \frac{\coth \beta_2}{\sinh^2 \beta_2} \bigg) \frac{1}{\sinh^2 \alpha_1} + \bigg( \frac{1}{\sinh^2 \beta_3} - \frac{1}{\sinh^2 \beta_2} \bigg)\frac{\coth \alpha_1}{\sinh^2 \alpha_1}  =
	    \frac{\coth \beta_3 - \coth \beta_2}{\sinh^2 \beta_2 \sinh^2 \beta_3} \,. 
	\end{equation}
\end{lemma}
\begin{proof}
	By multiplying the relation \eqref{eq: first order identity 2} by $\coth \beta_3 - \coth \beta_2$, and then using that \\ $\coth^2 \beta_3 - \coth^2 \beta_2 = \sinh^{-2} \beta_3 - \sinh^{-2} \beta_2$, we get
	\begin{equation} \label{eq: zeroth order identity 1 rearranged}
		   \bigg( \frac{\coth \beta_3 - \coth \beta_2}{\sinh^2 \beta_2} + \frac{\coth \beta_3 - \coth \beta_2}{\sinh^2 \beta_3} \bigg)\frac{1}{\sinh^2 \alpha_1} + 2 \bigg( \frac{1}{\sinh^2 \beta_3} - \frac{1}{\sinh^2 \beta_2} \bigg) \frac{\coth \alpha_1}{\sinh^2 \alpha_1} 
		= \frac{\coth \beta_3 - \coth \beta_2}{\sinh^2 \beta_2 \sinh^2 \beta_3} \,.
	\end{equation}
	Comparing relations \eqref{eq: zeroth order identity 1 rearranged} and \eqref{eq: zeroth order identity 1}, it remains to prove that 
	\begin{equation} \label{eq: eq: zeroth order identity 1 equivalent}
	   \bigg(\frac{\coth \beta_3}{\sinh^2 \beta_2} - \frac{\coth \beta_2}{\sinh^2 \beta_3} \bigg)\frac{1}{\sinh^2 \alpha_1}  + \bigg( \frac{1}{\sinh^2 \beta_3} - \frac{1}{\sinh^2 \beta_2} \bigg)\frac{\coth \alpha_1}{\sinh^2 \alpha_1}  = 0 \,.
	\end{equation}
	Note that since $\alpha_1 = \beta_2 + \beta_3$ we get
	\begin{equation*}
	 	\frac{\coth \alpha_1 - \coth \beta_2}{\sinh^2 \beta_3}  - \frac{\coth \alpha_1 - \coth \beta_3}{\sinh^2 \beta_2} = - \frac{\sinh \beta_3}{\sinh^2 \beta_3 \sinh \alpha_1 \sinh \beta_2} + \frac{\sinh \beta_2}{\sinh \beta_3 \sinh \alpha_1 \sinh^2 \beta_2} = 0 \,,
	\end{equation*}
	which implies that relation \eqref{eq: eq: zeroth order identity 1 equivalent} holds as required.
\end{proof}
\begin{lemma}\label{lemma: zeroth order identity 2}
	The following identity holds:
	\begin{equation} \label{eq: zeroth order identity 2}
	    \sum_{\sigma} \langle \beta_{\sigma(2)}, \beta_{\sigma(3)} \rangle \bigg( \frac{1}{\sinh^2 \beta_{\sigma(2)}} + \frac{1}{\sinh^2 \beta_{\sigma(3)}} \bigg) \frac{\coth\beta_{\sigma(1)}}{\sinh^2 \beta_{\sigma(1)}} 
	    = 2 \bigg( 2 + \sum_{j=1}^3 \frac{1}{\sinh^2 \beta_j} \bigg) X \,.
	\end{equation}
\end{lemma}

\begin{proof}
	Let us multiply the identity \eqref{eq: first order identity 4} in Lemma \ref{lemma: first order identity} by $\coth \beta_1$. It follows that 
	\begin{equation} \label{eq: zeroth order identity 2 LHS term}
	\begin{aligned}
	    &\bigg( \frac{1}{\sinh^2 \beta_2} + \frac{1}{\sinh^2 \beta_3} \bigg)\frac{\coth \beta_1}{\sinh^2 \beta_1}  = 
	    \frac{\coth \beta_1}{\sinh^2 \beta_2 \sinh^2 \beta_3} - 2 (\coth \beta_2 - \coth \beta_3) \frac{\coth^2 \beta_1}{\sinh^2 \beta_1} \\
	    &= \frac{\coth \beta_1}{\sinh^2 \beta_2 \sinh^2 \beta_3} + \frac{2\coth^2 \beta_1}{\sinh \beta_1 \sinh \beta_2 \sinh \beta_3} =
	    \frac{\coth \beta_1}{\sinh^2 \beta_2 \sinh^2 \beta_3} + 2 \omega^{-2} \bigg(1 + \frac{1}{\sinh^{2} \beta_1} \bigg)X \,.
	    \end{aligned}
	\end{equation} 
	We obtain two more variants of the relation \eqref{eq: zeroth order identity 2 LHS term} by applying $\pm \frac{\pi}{3}$ rotations and interchanging the $\beta$'s accordingly (see Remark \ref{rem: variants of identities}). By adding together the resulting three equalities, we get, with use of Lemma~\ref{lemma: alt expression for X}, that the left-hand side of the identity \eqref{eq: zeroth order identity 2} equals
	\begin{equation*}
	    \sum_{\sigma} \frac{\langle \beta_{\sigma(2)}, \beta_{\sigma(3)} \rangle \coth \beta_{\sigma(1)}}{\sinh^2 \beta_{\sigma(2)} \sinh^2 \beta_{\sigma(3)}} + 2 \bigg(3 + \sum_{j=1}^3 \frac{1}{\sinh^{2} \beta_j} \bigg) X = 
	    4 X + \bigg(\sum_{j=1}^3 \frac{2}{\sinh^2 \beta_j} \bigg) X \,.
	\end{equation*}
\end{proof}

\begin{lemma} \label{lemma: zeroth order identity 3}
	The following identity holds:
	\begin{equation} \label{eq: zeroth order id for gs term I}
		\begin{aligned} 
		    \sum_{\sigma} \langle \alpha_{\sigma(2)}, \alpha_{\sigma(3)} \rangle \bigg( \frac{1}{\sinh^2 \alpha_{\sigma(2)}} - \frac{1}{\sinh^2 \alpha_{\sigma(3)}} \bigg)\frac{\coth \alpha_{\sigma(1)}}{\sinh^2 \alpha_{\sigma(1)}}  = 0 \,.
		\end{aligned}
	\end{equation}
\end{lemma}

\begin{proof}
	Let us multiply both sides of \eqref{eq: zeroth order id for gs term I} by $-\frac13 \sinh^2 \alpha_1 \sinh^2 \alpha_2 \sinh^2 \alpha_3$. We need to prove that \\
	$(\sinh^2 \alpha_3 - \sinh^2 \alpha_2)\coth \alpha_1  - (\sinh^2 \alpha_1 - \sinh^2 \alpha_3)\coth \alpha_2  - (\sinh^2 \alpha_2 - \sinh^2 \alpha_1)\coth \alpha_3  = 0 $. We have
	\begin{equation} \label{eq: a term in zeroth order id for gs term I}
			(\sinh^2 \alpha_3 - \sinh^2 \alpha_2)\coth \alpha_1  = \sinh (\alpha_3 - \alpha_2)\cosh (\alpha_2 + \alpha_3) = \frac12 \sinh(2 \alpha_3) - \frac12 \sinh(2 \alpha_2) \,,
	\end{equation}
	since $\alpha_1 = \alpha_2 + \alpha_3$. By applying the rotations by $\pm \frac{\pi}{3}$ (see Remark \ref{rem: variants of identities}), we obtain from \eqref{eq: a term in zeroth order id for gs term I} expressions for $(\sinh^2 \alpha_1 - \sinh^2 \alpha_3)\coth \alpha_2 $ and $(\sinh^2 \alpha_2 - \sinh^2 \alpha_1)\coth \alpha_3$, which imply~the~statement.
\end{proof}

\section{The intertwining operator} \label{section: D}
     In this section we define the intertwining operator $\mathcal{D}$ given by formula \eqref{eq: D}, that is we define the corresponding functions $f_i$, $g_i$ and $h$. We also find the gradient and Laplacian  of these functions in a few lemmas in this section. This information is then used in Section \ref{section: intertwining relation} to prove the intertwining relation~\eqref{eq: intertwining relation}.
 
     We start with the following general lemma.
     \begin{lemma} \label{lemma: delta}
     	For any single-variable function $F$ and vectors $\alpha$, $\beta$, $\gamma$ such that $\langle \gamma, \gamma \rangle \neq 0$ we have 
     	\[
     		\partial_\alpha\partial_\beta \big(F \big(\langle \gamma, x \rangle \big) \big) = \frac{\langle \gamma, \alpha \rangle \langle \gamma, \beta \rangle }{\langle \gamma, \gamma \rangle} \Delta\big(F\big(\langle \gamma, x \rangle\big)\big) \,.
     	\] 
     \end{lemma}
     \begin{proof}
     	By the chain rule of differentiation, $\Delta(F(\langle \gamma, x \rangle)) = \langle \gamma, \gamma \rangle F''(\langle \gamma, x \rangle) $, while 	$\partial_\alpha\partial_\beta(F(\langle \gamma, x \rangle)) =$ \\ $\langle \gamma, \alpha \rangle \langle \gamma, \beta \rangle F''(\langle \gamma, x \rangle)$, where $F''$ denotes the second derivative of the function $F$.
     \end{proof}

   	In the expression \eqref{eq: D} for the operator $\mathcal{D} \,,$ let 
   	\begin{equation}
   		\begin{aligned} \label{def: fs}
   			f_j &= - (3m + 1) \langle \beta_j, \beta_j \rangle \coth\beta_j - \langle \beta_j, \beta_j \rangle \tanh\beta_j, 
   		\end{aligned}
   	\end{equation}
	where $j = 1,2,3$.     
    
    In the next lemma we calculate the gradient and Laplacian of the functions $f_j$.  
     \begin{lemma} \label{lemma: properties of fs}
     	The functions $f_j$ defined by expression $\eqref{def: fs}$, $j = 1,2,3$, satisfy the following relations:
     	\begin{enumerate}
     		\item \label{eq: nabla f} $\nabla(f_j) = \frac12 \widehat{u}_j\beta_j \,,$ (or, equivalently, $\partial_{\nabla(f_j)} = \frac12 \widehat{u}_j \partial_{\beta_j}$),
     		\item \label{eq: delta f} $-\Delta(f_j) + \widehat{u}_j f_j = \partial_{\beta_j}(u_j)$.
     	\end{enumerate}
     \end{lemma}
 
 	\begin{proof}
 	    Part \ref{eq: nabla f} follows from the equality  
 	    \[ 
 	    	\partial_i(f_j) = \bigg(\frac{(3m+1)\langle \beta_j, \beta_j\rangle}{\sinh^2 \beta_j} - \frac{\langle \beta_j, \beta_j\rangle}{\cosh^2 \beta_j} \bigg) \beta_j^{(i)} = \frac12 \widehat{u}_j \beta_j^{(i)}\,, 
 	 	\]
 	 	$i = 1, 2$, where $\beta_j = (\beta_j^{(1)}, \beta_j^{(2)})$ and we used the definition \eqref{eq: uhat's}.
 	 	
 	 	To establish property \ref{eq: delta f} we note that 
 		\[	
 			\Delta(f_j) = -\frac{2(3m+1)\langle \beta_j, \beta_j\rangle^2\coth \beta_j}{\sinh^2 \beta_j}  +  \frac{2\langle \beta_j, \beta_j\rangle^2\tanh \beta_j}{\cosh^2 \beta_j} \,. 
 		\]
 		Expanding and simplifying the product $\widehat{u}_j f_j$ yields
 		\[
 			\widehat{u}_j f_j = -\frac{2(3m+1)^2 \langle \beta_j, \beta_j\rangle^2 \coth \beta_j}{\sinh^2 \beta_j}  +  \frac{2\langle \beta_j, \beta_j\rangle^2\tanh \beta_j}{\cosh^2 \beta_j} \,.
 		\]
 		Therefore,
 		\[
 			-\Delta(f_j) + \widehat{u}_j f_j = -\frac{6m(3m+1)\langle \beta_j, \beta_j\rangle^2\coth \beta_j}{\sinh^2 \beta_j} =  \partial_{\beta_j}(u_j) \,,  
 		\] 
 		by relation \eqref{eq: u's}, as required.	
 	\end{proof}

   	For each $j = 1,2,3$, let $g_j$ in the operator \eqref{eq: D} be defined by
   	\begin{equation} \label{def: gs}
      	g_j = g_j^{(\RN{1})} + g_j^{(\RN{2})} + g_j^{(\RN{3})} \,, 
   	\end{equation}
   	where 
   	\begin{alignat}{2}
   	    &g_j^{(\RN{1})} &&= \prod_{k \neq j} f_k \,, \label{def: gI} \\
   	    &g_j^{(\RN{2})} &&= - \frac{\prod_{k \neq j} \langle \alpha_j, \beta_k \rangle}{\langle \alpha_j, \alpha_j \rangle} v_j \,, \label{def: gII} \\
   	    &g_j^{(\RN{3})} &&= - \frac{\prod_{k \neq j} \langle \beta_j, \beta_k \rangle}{\langle \beta_j, \beta_j \rangle} u_j \,,  \label{def: gIII} 
   	\end{alignat}
   	or, more explicitly,
   	\begin{alignat*}{2} 
   	    &g_1 &&= f_2 f_3 - \frac{9m(m+1)\omega^4}{\sinh^2 \alpha_1} + \frac{3m(3m+1)\omega^4}{\sinh^2 \beta_1} \,,  \\
   	    &g_2 &&= f_1 f_3 + \frac{9m(m+1)\omega^4}{\sinh^2 \alpha_2} - \frac{3m(3m+1)\omega^4}{\sinh^2 \beta_2} \,,  \\
   	    &g_3 &&= f_1 f_2 - \frac{9m(m+1)\omega^4}{\sinh^2 \alpha_3} + \frac{3m(3m+1)\omega^4}{\sinh^2 \beta_3} \,. 
   	\end{alignat*}
    
    In the next lemma we find gradients of the functions $g_j^{(\RN{2})}$, $g_j^{(\RN{3})}$.
 	\begin{lemma} \label{lemma: properties of gs part 0}
 		The functions $g_j^{(\RN{2})}$, $g_j^{(\RN{3})}$ defined by formulas \eqref{def: gII} and \eqref{def: gIII} satisfy the following relations for all $\sigma \in A_3$:
 		\begin{enumerate}
 			\item \label{eq: nabla gII no sum} $ \displaystyle 2\partial_{\nabla (g_{\sigma(1)}^{(\RN{2})} )} = 
 			-\partial_{\beta_{\sigma(2)}} \big(v_{\sigma(1)} \big)\partial_{\beta_{\sigma(3)}} -
 			\partial_{\beta_{\sigma(3)}} \big( v_{\sigma(1)} \big) \partial_{\beta_{\sigma(2)}} \,, $
 			\item \label{eq: nabla gIII no sum} $ \displaystyle 2\partial_{\nabla (g_{\sigma(1)}^{(\RN{3})})} = 
 			-\partial_{\beta_{\sigma(2)}} \big(u_{\sigma(1)} \big)\partial_{\beta_{\sigma(3)}} - 
 			\partial_{\beta_{\sigma(3)}} \big( u_{\sigma(1)} \big) \partial_{\beta_{\sigma(2)}} \,$.
 		\end{enumerate} 		
 	\end{lemma}
 
    \begin{proof}
    	We give proof for $\sigma = id$, the other cases are analogous. In the right-hand side of part \ref{eq: nabla gII no sum} we~have
    	\begin{equation*}
      	    - \partial_{\beta_2}(v_1) \partial_{\beta_3} - \partial_{\beta_3}(v_1) \partial_{\beta_2}   = 
    	    - \langle \alpha_1, \beta_2 \rangle v_1' \partial_{\beta_3} - \langle \alpha_1, \beta_3 \rangle v_1' \partial_{\beta_2} =
    	    - 3 v_1'\omega^2(\partial_{\beta_3} + \partial_{\beta_2}) = - 3 v_1' \omega^2\partial_{\alpha_1} \,,
    	\end{equation*} 
    	where $v_1'(x) = \frac{dV}{dz}|_{z=\langle \alpha_1, x \rangle}$, $ V(z) = m(m+1)\langle \alpha_1, \alpha_1\rangle \sinh^{-2}z $.
        And in the left-hand side of relation \ref{eq: nabla gII no sum} we get 
         \begin{equation*}
         	 2\partial_{\nabla(g_1^{(\RN{2})})} = -2 \frac{\langle \alpha_1, \beta_2 \rangle \langle \alpha_1, \beta_3 \rangle}{\langle \alpha_1, \alpha_1 \rangle} \partial_{\nabla(v_1)} = -2 \frac{\langle \alpha_1, \beta_2 \rangle \langle \alpha_1, \beta_3 \rangle}{\langle \alpha_1, \alpha_1 \rangle} v_1' \partial_{\alpha_1} = -3 \omega^2 v_1' \partial_{\alpha_1} \,,
         \end{equation*}
         so the two sides are equal.
        The proof of part \ref{eq: nabla gIII no sum} is similar.
    \end{proof}
 	
 	It will be useful to combine gradients of functions $g_j^{(\RN{1})}$, $g_j^{(\RN{2})}$, $g_j^{(\RN{3})}$ as in the following lemma.
 	\begin{lemma} \label{lemma: properties of gs part 1}
 		Functions $g_j^{(\RN{1})}$, $g_j^{(\RN{2})}$, $g_j^{(\RN{3})}$ defined by formulas \eqref{def: gI} -- \eqref{def: gIII} satisfy the following relations:
 		\begin{enumerate}
 			\item \label{eq: nabla gI} $ \displaystyle 2 \sum_{i=1}^{3} \partial_{\nabla(g_i^{(\RN{1})})} \partial_{\beta_i} = \sum_{\sigma} \bigg( \sum_{j  \neq \sigma(1)} \widehat{u}_j \bigg) f_{\sigma(1)} \partial_{\beta_{\sigma(2)}}\partial_{\beta_{\sigma(3)}} \,,$ 
 			\item \label{eq: nabla gII} $\displaystyle 2 \sum_{i=1}^{3} \partial_{\nabla(g_i^{(\RN{2})})} \partial_{\beta_i} = - \sum_{\sigma} \bigg(\partial_{\beta_{\sigma(2)}} \big(v_{\sigma(1)} \big)\partial_{\beta_{\sigma(3)}} + 
 			\partial_{\beta_{\sigma(3)}} \big( v_{\sigma(1)} \big) \partial_{\beta_{\sigma(2)}} \bigg) \partial_{\beta_{\sigma(1)}}
 			\\ \null \qquad\qquad\qquad\quad = - \sum_{\sigma} \partial_{\beta_{\sigma(1)}} \bigg(\sum_{j \neq \sigma(1)} v_j \bigg) \partial_{\beta_{\sigma(2)}}\partial_{\beta_{\sigma(3)}}
 			= - \sum_{\sigma} \partial_{\beta_{\sigma(1)}} \bigg(\sum_{j=1}^3 v_j \bigg) \partial_{\beta_{\sigma(2)}}\partial_{\beta_{\sigma(3)}} \,,$
 			\item \label{eq: nabla gIII}$\displaystyle 2 \sum_{i=1}^{3} \partial_{\nabla(g_i^{(\RN{3})})} \partial_{\beta_i} = - \sum_{\sigma} \bigg(\partial_{\beta_{\sigma(2)}} \big( u_{\sigma(1)} \big) \partial_{\beta_{\sigma(3)}} + 
 			\partial_{\beta_{\sigma(3)}} \big(u_{\sigma(1)} \big)\partial_{\beta_{\sigma(2)}} \bigg) \partial_{\beta_{\sigma(1)}}
 			\\ \null \qquad\qquad\qquad\quad = - \sum_{\sigma} \partial_{\beta_{\sigma(1)}} \bigg(\sum_{j \neq \sigma(1)} u_j \bigg) \partial_{\beta_{\sigma(2)}}\partial_{\beta_{\sigma(3)}} \,.$
 		\end{enumerate}
 	\end{lemma}
 
    \begin{proof}
    	In order to prove part \ref{eq: nabla gI}, we note that by the definition \eqref{def: gI} we have
    	\begin{equation}
    		\begin{aligned} \label{eq: nabla gI step 1}
	    		2 \sum_{i=1}^{3} \partial_{\nabla(g_i^{(\RN{1})})} \partial_{\beta_i} &= 2 \sum_{\sigma} \partial_{\nabla(f_{\sigma(2)} f_{\sigma(3)})} 
	    		\partial_{\beta_{\sigma(1)}}  \\
	    		&= 2 \sum_{\sigma} f_{\sigma(3)} \partial_{\nabla(f_{\sigma(2)})}  \partial_{\beta_{\sigma(1)}} +
	    		2 \sum_{\sigma} f_{\sigma(2)} \partial_{\nabla(f_{\sigma(3)})}  \partial_{\beta_{\sigma(1)}} \,.
    		\end{aligned}
    	\end{equation}
        Substituting the result of Lemma~\ref{lemma: properties of fs}~part~\ref{eq: nabla f} for $\partial_{\nabla(f_j)}$ into the expression \eqref{eq: nabla gI step 1} we obtain
    	\[
    		 \sum_{\sigma} f_{\sigma(3)} \widehat{u}_{\sigma(2)} \partial_{\beta_{\sigma(2)}}  \partial_{\beta_{\sigma(1)}} +
    		 \sum_{\sigma} f_{\sigma(2)} \widehat{u}_{\sigma(3)} \partial_{\beta_{\sigma(3)}}  \partial_{\beta_{\sigma(1)}} \,,
    	\]
    	which equals the right-hand side of \ref{eq: nabla gI}.
        
       The equalities \ref{eq: nabla gII} and \ref{eq: nabla gIII} follow from Lemma \ref{lemma: properties of gs part 0}.
	\end{proof}
    
    In the next lemma we deal with combining gradients of functions $f_j$ and $g_j^{(\RN{1})}$, $g_j^{(\RN{2})}$, $g_j^{(\RN{3})}$.
    \begin{lemma} \label{lemma: properties of gs part 2}
        The functions $g_j^{(\RN{1})}$, $g_j^{(\RN{2})}$, $g_j^{(\RN{3})}$ defined by formulas $\eqref{def: gI}$ -- $\eqref{def: gIII}$ satisfy also the following relations:
    	\begin{enumerate}
    		\item \label{eq: nabla gI dot nabla f} $\displaystyle \sum_{i=1}^{3} \langle \nabla(f_i), \nabla(g_i^{(\RN{1})}) \rangle = 
    		\frac12 \sum_{\sigma} \langle \beta_{\sigma(2)}, \beta_{\sigma(3)} \rangle \widehat{u}_{\sigma(2)} \widehat{u}_{\sigma(3)} f_{\sigma(1)} \,,$
    		\item \label{eq: nabla gII dot nabla f} $\displaystyle \sum_{i=1}^{3} \langle \nabla(f_i), \nabla(g_i^{(\RN{2})}) \rangle = 0 \,,$
    		\item \label{eq: nabla gIII dot nabla f} $\displaystyle 2 \sum_{i=1}^{3} \langle \nabla(f_i), \nabla(g_i^{(\RN{3})}) \rangle = 
    		\sum_{i=1}^3 \widehat{u}_i \partial_{\beta_i} (g_i^{(\RN{3})}) \,.$
    	\end{enumerate}
    \end{lemma}
    
    \begin{proof}
    	The left-hand side of \ref{eq: nabla gI dot nabla f} can be expanded using the product rule as 
    	\begin{align*}
    		\sum_{i=1}^{3} \langle \nabla(f_i), \nabla(g_i^{(\RN{1})}) \rangle &= \sum_{\sigma} \langle \nabla{f_{\sigma(1)}}, \nabla(f_{\sigma(2)} f_{\sigma(3)}) \rangle \\
    		&= \sum_{\sigma} \langle \nabla(f_{\sigma(1)}), \nabla(f_{\sigma(2)}) \rangle f_{\sigma(3)} 
    		+ \sum_{\sigma} \langle \nabla(f_{\sigma(1)}), \nabla(f_{\sigma(3)}) \rangle f_{\sigma(2)}  \\
    		&= 2 \sum_{\sigma} \langle \nabla(f_{\sigma(2)}), \nabla(f_{\sigma(3)}) \rangle f_{\sigma(1)} \,,
    	\end{align*} 
        and the result follows by an application of Lemma \ref{lemma: properties of fs} part \ref{eq: nabla f}.
    	 
    	The relation \ref{eq: nabla gII dot nabla f} holds because $\nabla(f_i)$ is proportional to $\beta_i$, while $\nabla(g_i^{(\RN{2})})$ is proportional to $\alpha_i \,$, and $\langle \alpha_i, \beta_i \rangle = 0 \,$, for all  $i=1,2,3$.
    
	    Further, we have by Lemma \ref{lemma: properties of fs} part \ref{eq: nabla f} that
	    \begin{align*}
	    	2\sum_{i=1}^{3} \langle \nabla(f_i), \nabla(g_i^{(\RN{3})}) \rangle = 
	    	\sum_{i=1}^{3} \widehat{u}_i \langle \beta_i, \nabla(g_i^{(\RN{3})}) \rangle =
	    	\sum_{i=1}^3 \widehat{u}_i \partial_{\beta_i} (g_i^{(\RN{3})}) \,,
	    \end{align*}
	    which proves identity \ref{eq: nabla gIII dot nabla f}.
	\end{proof} 	
    
    In the next Lemmas \ref{lemma: properties of gs part 3-} and \ref{lemma: properties of gs part 3} we calculate and rearrange Laplacians of functions $g_j^{(\RN{1})}$, $g_j^{(\RN{2})}$, $g_j^{(\RN{3})}$. 
    \begin{lemma}\label{lemma: properties of gs part 3-}
    	The functions $g_j^{(\RN{2})}$ and $g_j^{(\RN{3})}$  defined by formulas \eqref{def: gII} and \eqref{def: gIII} satisfy the following relations for all $i = 1,2,3$:
    	\begin{enumerate}
    		\item \label{eq: delta gII no sum} $\displaystyle \Delta(g_i^{(\RN{2})}) = - \prod_{k \neq i}
    		\partial_{\beta_k} v_i = - \prod_{k \neq i}
    		\partial_{\beta_k} \sum_{j=1}^3 v_j  ,$
    		\item \label{eq: delta gIII no sum} $\displaystyle \Delta(g_i^{(\RN{3})}) = - \prod_{k \neq i}
    		\partial_{\beta_k} u_i  \,$. 
    	\end{enumerate}
    \end{lemma}
    \begin{proof}
Statement \ref{eq: delta gII no sum} follows by formula \eqref{def: gII} and Lemma \ref{lemma: delta}.
    	Similarly, property \ref{eq: delta gIII no sum} follows directly from Lemma \ref{lemma: delta} and formula \eqref{def: gIII}. 
    \end{proof}
    
	\begin{lemma} \label{lemma: properties of gs part 3}
		Functions $g_j^{(\RN{1})}$, $g_j^{(\RN{2})}$, $g_j^{(\RN{3})}$ defined by formulas \eqref{def: gI} -- \eqref{def: gIII} satisfy the following relations:
		\begin{enumerate}
			\item \label{eq: delta gI} $\displaystyle \sum_{i=1}^3 \Delta(g_i^{(\RN{1})}) \partial_{\beta_i} = 
			\sum_{i=1}^3 \bigg( \sum_{j  \neq i} \widehat{u}_j \bigg) g_i^{(\RN{1})} \partial_{\beta_i}
			+ \frac12 \sum_{\sigma} \langle \beta_{\sigma(2)}, \beta_{\sigma(3)} \rangle \widehat{u}_{\sigma(2)} \widehat{u}_{\sigma(3)} \partial_{\beta_{\sigma(1)}}
			\\ \null \qquad\qquad\qquad\quad- \sum_{\sigma} f_{\sigma(1)} \bigg(\partial_{\beta_{\sigma(2)}}(u_{\sigma(2)})\partial_{\beta_{\sigma(3)}} + 
			\partial_{\beta_{\sigma(3)}}(u_{\sigma(3)})\partial_{\beta_{\sigma(2)}} \bigg) \,,$ 
			\item \label{eq: delta gII} $\displaystyle \sum_{i=1}^3 \Delta(g_i^{(\RN{2})}) \partial_{\beta_i} = -\sum_{\sigma}
			\partial_{\beta_{\sigma(2)}} \partial_{\beta_{\sigma(3)}} \bigg(\sum_{j=1}^3 v_j \bigg) \partial_{\beta_{\sigma(1)}} \,,$ 
			\item \label{eq: delta gIII} $\displaystyle \sum_{i=1}^3 \Delta(g_i^{(\RN{3})}) \partial_{\beta_i} = -\sum_{\sigma}
			\partial_{\beta_{\sigma(2)}} \partial_{\beta_{\sigma(3)}} \big(u_{\sigma(1)} \big) \partial_{\beta_{\sigma(1)}} \,.$ 
		\end{enumerate}
	\end{lemma}

    \begin{proof}
    	Let us first consider $\Delta(g_1^{(\RN{1})})$. By Lemma \ref{lemma: properties of fs}, we have
    	\begin{equation} \label{eq: delta g1I}
    		\begin{aligned}
    		\Delta(g_1^{(\RN{1})}) &= \Delta(f_2 f_3) = \Delta(f_2) f_3 + 2 \langle \nabla(f_2), \nabla(f_3) \rangle + \Delta(f_3) f_2  \\
    		&= \bigg( \widehat{u}_2 f_2 - \partial_{\beta_2}(u_2) \bigg) f_3 + \frac12 \widehat{u}_2 \widehat{u}_3 \langle \beta_2, \beta_3 \rangle 
    		+ \bigg( \widehat{u}_3 f_3 - \partial_{\beta_3}(u_3) \bigg) f_2 \\
    		&= \bigg( \sum_{j  \neq 1} \widehat{u}_j \bigg) g_1^{(\RN{1})} + \frac12 \widehat{u}_2 \widehat{u}_3 \langle \beta_2, \beta_3 \rangle - \big(\partial_{\beta_2}(u_2) f_3 + \partial_{\beta_3}(u_3) f_2 \big) \,.
    		\end{aligned}
    	\end{equation}
    	By multiplying \eqref{eq: delta g1I} by $\partial_{\beta_1}$, and adding it with similar expressions for $\Delta(g_2^{(\RN{1})}) \partial_{\beta_2}$ and
    	$\Delta(g_3^{(\RN{1})}) \partial_{\beta_3}$, we obtain property \ref{eq: delta gI}.
        
        Properties \ref{eq: delta gII} and \ref{eq: delta gIII} follow from Lemma \ref{lemma: properties of gs part 3-} parts \ref{eq: delta gII no sum} and \ref{eq: delta gIII no sum}, respectively, by multiplying these equalities by $\partial_{\beta_i}$ and summing them  up over $i = 1,2,3$.
    \end{proof}

  	 Let $h$ in the operator \eqref{eq: D} be defined by
   	\begin{equation}
   		h = h^{(\RN{1})} + h^{(\RN{2})} + h^{(\RN{3})} + h^{(\RN{4})} \,,
   	\end{equation}
   	where 
   	\begin{alignat}{2}
    	&h^{(\RN{1})} &&= f_1 f_2 f_3 \,, \label{def: hI} \\
    	&h^{(\RN{2})} &&= \sum_{i=1}^3 f_i \big( g_i^{(\RN{2})} + g_i^{(\RN{3})} \big) \,, \label{def: hII} \\
    	&h^{(\RN{3})} &&= \sum_{i=1}^3 \partial_{\beta_i}(g_i^{(\RN{3})}) = - \sum_{i=1}^3 \frac{\prod_{k \neq i} \langle \beta_i, \beta_k \rangle}{\langle \beta_i, \beta_i \rangle}  \partial_{\beta_i}(u_i) \label{def: hIII}  \\
    	& && = - \frac{12m(3m+1)\omega^6}{\sinh^2 \beta_1} \coth \beta_1 + \frac{12m(3m+1)\omega^6}{\sinh^2 \beta_2} \coth \beta_2 - \frac{12m(3m+1)\omega^6}{\sinh^2 \beta_3} \coth \beta_3 \,, \nonumber \\
    	&h^{(\RN{4})} &&= -3m(3m+1)\omega^{-2}\prod_{i=1}^3 \langle \beta_i, \beta_i\rangle X - 4(3m+1)\omega^{-2} \prod_{i=1}^3 \langle \beta_i, \beta_i\rangle Y. \label{def: hIV} 
   	\end{alignat}
   
    In the next Lemmas \ref{lemma: properties of h part 1}, \ref{lemma: properties of h part 2} we calculate gradients and Laplacians of the functions $h^{(\RN{1})}$, $h^{(\RN{2})}$, $h^{(\RN{3})}$.
   	\begin{lemma} \label{lemma: properties of h part 1}
   		The functions $h^{(\RN{1})}$, $h^{(\RN{2})}$, $h^{(\RN{3})}$ defined by formulas \eqref{def: hI} -- \eqref{def: hIV} satisfy the following relations:
   		\begin{enumerate}
   			\item \label{eq: nabla of hI & II} $\displaystyle 2 \partial_{\nabla(h^{(\RN{1})} + h^{(\RN{2})})} = 
   			\sum_{i=1}^3 \widehat{u}_i g_i \partial_{\beta_i} - \sum_{\sigma} f_{\sigma(1)} \bigg(\partial_{\beta_{\sigma(2)}}\big(v_{\sigma(1)} + u_{\sigma(1)} \big)\partial_{\beta_{\sigma(3)}} + \partial_{\beta_{\sigma(3)}} \big(v_{\sigma(1)} + u_{\sigma(1)} \big)\partial_{\beta_{\sigma(2)}} \bigg)$, 
   			\item \label{eq: nabla of hIII} $\displaystyle 2 \partial_{\nabla(h^{(\RN{3})})} = 
   			- \sum_{\sigma}	\partial_{\beta_{\sigma(2)}}\partial_{\beta_{\sigma(3)}} \bigg(\sum_{j \neq \sigma(1)} u_j \bigg)  \partial_{\beta_{\sigma(1)}} $.
   		\end{enumerate}
   	\end{lemma}
    
    \begin{proof}
    	We have that 
    	\[
    		\partial_j(h^{(\RN{1})}) = \partial_j(f_1) f_2 f_3 + \partial_j(f_2) f_1 f_3 + \partial_j(f_3) f_1 f_2 = \sum_{i=1}^3 \partial_j(f_i) g_i^{(\RN{1})} \,,
    	\]
    	therefore by Lemma \ref{lemma: properties of fs} part \ref{eq: nabla f},
    	\begin{equation} \label{eq: nabla hI}
    		2 \partial_{\nabla(h^{(\RN{1})})} = \sum_{i=1}^3 2 g_i^{(\RN{1})} \partial_{\nabla(f_i)} = \sum_{i=1}^3 \widehat{u}_i g_i^{(\RN{1})} \partial_{\beta_i} \,.
    	\end{equation}
    	On the other hand,
    	\begin{alignat}{2} \label{eq: nabla hII}
    		2 \partial_{\nabla(h^{(\RN{2})})} = 2 \sum_{i=1}^3 \partial_{\nabla\big( f_i (g_i^{(\RN{2})} + g_i^{(\RN{3})}) \big)} 
    		= 2 \sum_{i=1}^3 \big(g_i^{(\RN{2})} + g_i^{(\RN{3})} \big) \partial_{\nabla(f_i)} + 2 \sum_{i=1}^3 f_i \bigg(\partial_{\nabla(g_i^{(\RN{2})})} +  \partial_{\nabla(g_i^{(\RN{3})})} \bigg)  \,.
    	\end{alignat} 
    	By Lemma \ref{lemma: properties of fs} part \ref{eq: nabla f} and Lemma \ref{lemma: properties of gs part 0} we can rearrange the expression \eqref{eq: nabla hII} as
    	\begin{equation} \label{eq: nabla hII rearranged}
	    	 \sum_{i=1}^3 \widehat{u}_i \big(g_i^{(\RN{2})} + g_i^{(\RN{3})} \big) \partial_{\beta_i} 
	    	 - \sum_{\sigma} f_{\sigma(1)} \bigg( \partial_{\beta_{\sigma(2)}} \big( v_{\sigma(1)} + u_{\sigma(1)} \big) \partial_{\beta_{\sigma(3)}}  +  \partial_{\beta_{\sigma(3)}} \big( v_{\sigma(1)} + u_{\sigma(1)} \big) \partial_{\beta_{\sigma(2)} } \bigg) \,.
    	\end{equation}
    	The statement \ref{eq: nabla of hI & II} follows by adding up equalities \eqref{eq: nabla hI} and \eqref{eq: nabla hII rearranged}.
    	
    	In the right-hand side of statement \ref{eq: nabla of hIII}, the coefficient at $\partial_{\beta_{1}}$ is equal to
    	\begin{equation} \label{eq: nabla of hIII RHS}
    		- \partial_{\beta_2} \partial_{\beta_3}(u_2 + u_3) = -24m(3m+1) \omega^6 \bigg( \frac{2 \coth^2 \beta_2}{\sinh^2 \beta_2} + \frac{1}{\sinh^4 \beta_2} +  \frac{2 \coth^2 \beta_3}{\sinh^2 \beta_3} + \frac{1}{\sinh^4 \beta_3} \bigg) \,.
    	\end{equation}
    	In the left-hand side of statement \ref{eq: nabla of hIII} one can check that $2 \partial_{\nabla(h^{(\RN{3})})}$ is equal to
    	\begin{equation} \label{eq: nabla of hIII LHS}
	    	\begin{aligned}
		    	24m(3m+1) \omega^6\times\\
 \bigg( \bigg( \frac{2\coth^2 \beta_1}{\sinh^2 \beta_1} + \frac{1}{\sinh^4 \beta_1} \bigg) \partial_{\beta_1} - 
		    	\bigg( \frac{2\coth^2 \beta_2}{\sinh^2 \beta_2} + \frac{1}{\sinh^4 \beta_2} \bigg) \partial_{\beta_2} 
		    	+ \bigg( \frac{2\coth^2 \beta_3}{\sinh^2 \beta_3} + \frac{1}{\sinh^4 \beta_3} \bigg) \partial_{\beta_3} \bigg) \,.
	    	\end{aligned}
    	\end{equation}
    	Let us substitute in expression \eqref{eq: nabla of hIII LHS} $\partial_{\beta_1} = \partial_{\beta_2} - \partial_{\beta_3}$, $\partial_{\beta_2} = \partial_{\beta_1} + \partial_{\beta_3}$, and $\partial_{\beta_3} = \partial_{\beta_2} - \partial_{\beta_1}$. Then one can see that the coefficient at $\partial_{\beta_1}$ equals expression \eqref{eq: nabla of hIII RHS}. Similarly, the coefficients at $\partial_{\beta_2}$ and $\partial_{\beta_3}$ also match on both sides of equality \ref{eq: nabla of hIII}.
    \end{proof}

	\begin{lemma} \label{lemma: properties of h part 2}
		Functions $h^{(\RN{1})}$, $h^{(\RN{2})}$, $h^{(\RN{3})}$ given by formulas \eqref{def: hI} -- \eqref{def: hIII} satisfy the following relations:
		\begin{enumerate}
			\item \label{eq: delta of hI && II} $\displaystyle \Delta(h^{(\RN{1})} + h^{(\RN{2})}) = 
			\sum_{i=1}^3 \widehat{u}_i f_i g_i  - \sum_{i=1}^3 \partial_{\beta_i}(u_i) g_i
			+ \frac12 \sum_{\sigma} \langle \beta_{\sigma(2)}, \beta_{\sigma(3)} \rangle \widehat{u}_{\sigma(2)} \widehat{u}_{\sigma(3)} f_{\sigma(1)}
			\\ \null \quad\qquad\qquad\qquad + \sum_{i=1}^3 \widehat{u}_i \partial_{\beta_i} \big( g_i^{(\RN{3})} \big) -\sum_{\sigma} f_{\sigma(1)}
			\partial_{\beta_{\sigma(2)}} \partial_{\beta_{\sigma(3)}} \big( v_{\sigma(1)} + u_{\sigma(1)} \big) $, 
			
			\item \label{eq: delta of hIII} $\displaystyle \Delta(h^{(\RN{3})}) = 
			- \partial_{\beta_{1}} \partial_{\beta_{2}} \partial_{\beta_{3}} \bigg(\sum_{j=1}^3 u_j\bigg) $. 
		\end{enumerate}
	\end{lemma}
    
    \begin{proof}
    	Firstly, by Lemma \ref{lemma: properties of fs} part \ref{eq: delta f} and Lemma \ref{lemma: properties of gs part 2} part \ref{eq: nabla gI dot nabla f} we have
    	\begin{equation} \label{eq: delta hI}
	    	\begin{aligned}
	    	\Delta(h^{(\RN{1})}) &= \Delta(f_1 f_2 f_3) = \sum_{i=1}^3 \Delta(f_i) g_i^{(\RN{1})} + \sum_{i=1}^3 \langle \nabla(f_i), \nabla(g_i^{(\RN{1})}) \rangle \\
	    	&= \sum_{i=1}^3 \widehat{u}_i f_i g_i^{(\RN{1})} - \sum_{i=1}^3 \partial_{\beta_i}(u_i) g_i^{(\RN{1})}
	    	+ \frac12 \sum_{\sigma} \widehat{u}_{\sigma(2)} \widehat{u}_{\sigma(3)} \langle \beta_{\sigma(2)}, \beta_{\sigma(3)} \rangle f_{\sigma(1)} \,.
	    	\end{aligned}
    	\end{equation}
    	Secondly, by Lemma \ref{lemma: properties of fs} part \ref{eq: delta f}, by Lemma \ref{lemma: properties of gs part 2} parts \ref{eq: nabla gII dot nabla f} and \ref{eq: nabla gIII dot nabla f}, and Lemma \ref{lemma: properties of gs part 3-} we have
    	\begin{equation} \label{eq: delta hII}
    		\begin{aligned}
    		\Delta(h^{(\RN{2})}) &= \Delta \bigg( \sum_{i=1}^3 f_i \big( g_i^{(\RN{2})} + g_i^{(\RN{3})} \big) \bigg) \\
    		&= \sum_{i=1}^3 \Delta(f_i) \big( g_i^{(\RN{2})} + g_i^{(\RN{3})} \big) + 2 \sum_{i=1}^3 \langle \nabla(f_i), \nabla \big( g_i^{(\RN{2})} + g_i^{(\RN{3})} \big) \rangle + \sum_{i=1}^3 f_i \Delta \big( g_i^{(\RN{2})} + g_i^{(\RN{3})} \big) \\
    		&= \sum_{i=1}^3 \widehat{u}_i f_i \big( g_i^{(\RN{2})} + g_i^{(\RN{3})} \big) - \sum_{i=1}^3 \partial_{\beta_i}(u_i) \big( g_i^{(\RN{2})} + g_i^{(\RN{3})} \big) \\
    		&\quad + \sum_{i=1}^3 \widehat{u}_i \partial_{\beta_i} \big( g_i^{(\RN{3})} \big) - \sum_{\sigma} f_{\sigma(1)}
    		\partial_{\beta_{\sigma(2)}} \partial_{\beta_{\sigma(3)}} \big( v_{\sigma(1)} + u_{\sigma(1)} \big) \,.
    		\end{aligned}
    	\end{equation}
    	The statement \ref{eq: delta of hI && II} follows by adding the equalities \eqref{eq: delta hI} and \eqref{eq: delta hII}.
    	
    	By Lemma \ref{lemma: delta} for any $j$ we have
    	\begin{equation} \label{eq: delta of hjIII}
    		\Delta\bigg( -\frac{\prod_{k \neq j} \langle \beta_j, \beta_k \rangle}{\langle \beta_j, \beta_j \rangle}  \partial_{\beta_j}(u_j) \bigg) =
    		- \bigg( \prod_{k \neq j} \partial_{\beta_k} \bigg) \partial_{\beta_j}(u_j) = - \partial_{\beta_{1}} \partial_{\beta_{2}} \partial_{\beta_{3}} (u_j) \,.
    	\end{equation} 
    	We get result \ref{eq: delta of hIII} by summing equalities \eqref{eq: delta of hjIII} over $j = 1,2,3$.
    \end{proof}

    \section{Proof of the intertwining relation} \label{section: intertwining relation}
    Let $A = A(x, \partial_1, \partial_2)$ be a differential operator of order $N$. Then $A$ can be represented as 
    \[
    	A = \sum_{k = 0}^{N} A^{(k)} \,, \quad \text{with } A^{(k)} = \sum_{i+j=k} a_{ij}(x) \partial_1^i \partial_2^j 
    \]
    for some functions $a_{ij}(x)$ 
    so $A^{(k)}$ denotes the $k$-th order part of $A$. That is $A^{(k)}$ is the sum of all terms in $A$ that contain exactly $k$ derivatives when all the derivatives are put on the right.
    
    Both operators $H  \mathcal{D}$ and $\mathcal{D}  H_0$ have order 5. It is easy to see that the respective terms of orders~5 and 4 in both operators are the same. We are going to show that this is also true for lower orders.
    \begin{proposition}
    	The third order terms in the intertwining relation \eqref{eq: intertwining relation} satisfy $$(H  \mathcal{D})^{(3)}  = (\mathcal{D} H_0)^{(3)} \,.$$
    \end{proposition}

    \begin{proof}
    	We have
    	\[
    		(H \mathcal{D} -   \mathcal{D} H_0)^{(3)} = -2 \sum_{\sigma} \partial_{\nabla(f_{\sigma(1)})} \partial_{\beta_{\sigma(2)}} \partial_{\beta_{\sigma(3)}} + \bigg(\sum_{j=1}^3 \widehat{u}_j \bigg) \partial_{\beta_1} \partial_{\beta_2} \partial_{\beta_3} \,.
    	\]
    	By Lemma \ref{lemma: properties of fs} part \ref{eq: nabla f} we get
    	\[
    		\sum_{\sigma} \partial_{\nabla(f_{\sigma(1)})} \partial_{\beta_{\sigma(2)}} \partial_{\beta_{\sigma(3)}} = 
    		\frac12 \sum_{\sigma} \widehat{u}_{_{\sigma(1)}} \partial_{\beta_{\sigma(1)}} \partial_{\beta_{\sigma(2)}} \partial_{\beta_{\sigma(3)}} =
    		\frac12 \bigg(\sum_{j=1}^3 \widehat{u}_j \bigg) \partial_{\beta_1} \partial_{\beta_2} \partial_{\beta_3} \,,
    	\]
    	and the statement follows.
    \end{proof}

	\begin{proposition}
		The second order terms in the intertwining relation \eqref{eq: intertwining relation} satisfy $$(H  \mathcal{D})^{(2)}  = (\mathcal{D}  H_0)^{(2)} \,.$$
	\end{proposition}

	\begin{proof}
		We have that 
		\begin{align*}
			(H  \mathcal{D} -  \mathcal{D} H_0)^{(2)} &=
			- \sum_{\sigma} \Delta(f_{_{\sigma(1)}}) \partial_{\beta_{\sigma(2)}} \partial_{\beta_{\sigma(3)}}
			- 2 \sum_{i=1}^3 \partial_{\nabla(g_i)} \partial_{\beta_i}  \\
			&\quad+ \sum_{\sigma} \sum_{j=1}^3 \widehat{u}_j f_{_{\sigma(1)}} \partial_{\beta_{\sigma(2)}} \partial_{\beta_{\sigma(3)}}
			- \sum_{\sigma} \partial_{\beta_{\sigma(1)}} \bigg( \sum_{j=1}^3(v_j + u_j) \bigg) \partial_{\beta_{\sigma(2)}} \partial_{\beta_{\sigma(3)}} \,,
		\end{align*} which is zero by applying Lemma \ref{lemma: properties of fs} part \ref{eq: delta f} and adding equalities from all three parts of Lemma~\ref{lemma: properties of gs part 1}. 
	\end{proof}

    The next lemma will be useful for dealing with the first and zero order terms in the intertwining relation.
    \begin{lemma} \label{lemma: first and zeroth order}
    	For any $\sigma \in A_3$,
    	\begin{align}
\label{newlemnum}
    		&-\frac12 \langle \beta_{\sigma(2)}, \beta_{\sigma(3)} \rangle \widehat{u}_{\sigma(2)} \widehat{u}_{\sigma(3)}
    		+ \bigg( \sum_{j  \neq \sigma(1)} \widehat{u}_j \bigg) \big(g_{\sigma(1)}^{(\RN{2})} + g_{\sigma(1)}^{(\RN{3})} \big) 
    		- \bigg( f_{\sigma(2)} \partial_{\beta_{\sigma(3)}} + f_{\sigma(3)} \partial_{\beta_{\sigma(2)}} \bigg)(v_{\sigma(1)} + u_{\sigma(1)}) \nonumber \\
    		&\qquad = -\frac{48 m(3m+1) \omega^4 \langle \beta_{\sigma(2)}, \beta_{\sigma(3)} \rangle}{\sinh^2 \beta_{\sigma(2)} \sinh^2 \beta_{\sigma(3)}} - \frac{128 (3m+1)\omega^4 \langle \beta_{\sigma(2)}, \beta_{\sigma(3)} \rangle}{\sinh^2 2\beta_{\sigma(2)} \sinh^2 2\beta_{\sigma(3)}} \,.
    	\end{align}
    \end{lemma}

    \begin{proof}
        Firstly we let $\sigma = id$. 
    	By the identities \eqref{eq: first order identity 1} and~\eqref{eq: first order identity 2} in Lemma \ref{lemma: first order identity} we get
    	\begin{equation} \label{eq: first order Db1 I}
    	\bigg( \sum_{j  \neq 1} \widehat{u}_j \bigg) g_1^{(\RN{2})} -\big( f_2 \partial_{\beta_3} + f_3 \partial_{\beta_2}\big)(v_1) =
    	- 36 m (m+1)\omega^6  \bigg( \frac{3m+1}{\sinh^2 \beta_2 \sinh^2 \beta_3} + \frac{1}{\cosh^2 \beta_2 \cosh^2 \beta_3} \bigg) \,.
    	\end{equation}
    	Similarly, by the identities \eqref{eq: first order identity 3} and \eqref{eq: first order identity 4} in Lemma \ref{lemma: first order identity} we get
    	\begin{equation} \label{eq: first order Db1 II}
    	\bigg( \sum_{j  \neq 1} \widehat{u}_j \bigg) g_1^{(\RN{3})} -\big(f_2 \partial_{\beta_3} + f_3 \partial_{\beta_2}\big)(u_1) = 
    	12m(3 m+1) \omega^6 \bigg( \frac{3m+1}{\sinh^2 \beta_2 \sinh^2 \beta_3} + \frac{1}{\cosh^2 \beta_2 \cosh^2 \beta_3} \bigg) \,.
    	\end{equation}
    	Note that the sum of expressions \eqref{eq: first order Db1 I}, \eqref{eq: first order Db1 II} divided by $\omega^6$ together with the  term $-\frac12 \widehat{u}_{2} \widehat{u}_{3} $ divided by $\omega^4$ equals
    	\begin{equation*}
	    	\begin{aligned}
		    	& -\frac{48m(3m+1)}{\sinh^2 \beta_2 \sinh^2 \beta_3}  
		    	-8(3m+1) \bigg( \frac{1}{\sinh^2 \beta_2} - \frac{1}{\cosh^2 \beta_2} \bigg) \bigg( \frac{1}{\sinh^2 \beta_3} - \frac{1}{\cosh^2 \beta_3} \bigg) \\
		    	&= -\frac{48m(3m+1)}{\sinh^2 \beta_2 \sinh^2 \beta_3} - \frac{128 (3m+1)}{\sinh^2 2 \beta_2 \sinh^2 2\beta_3} \,,
	    	\end{aligned}
    	\end{equation*}
which is the right-hand side of the equality \eqref{newlemnum} divided by $\omega^6$    	as required.
    	The cases $\sigma \ne id$ follow from versions of \eqref{eq: first order identity 1} -- \eqref{eq: first order identity 4} obtained by rotating vectors (see Remark~\ref{rem: variants of identities}). 
    \end{proof}

	\begin{proposition} \label{prop: first order}
		The first order terms in the intertwining relation \eqref{eq: intertwining relation} satisfy $$(H  \mathcal{D})^{(1)}  = (\mathcal{D} H_0)^{(1)} \,.$$
	\end{proposition}
	
	\begin{proof}
		We have that 
		\begin{equation} \label{eq: first order}
			\begin{aligned}
				(H  \mathcal{D} - \mathcal{D} H_0)^{(1)} &= 
				- \sum_{i=1}^3 \Delta(g_i) \partial_{\beta_i} - 2 \partial_{\nabla(h)} + \sum_{i=1}^3 \bigg( \sum_{j=1}^3 \widehat{u}_j \bigg) g_i \partial_{\beta_i} - \sum_{\sigma} \partial_{\beta_{\sigma(2)}} \partial_{\beta_{\sigma(3)}} \bigg(\sum_{j=1}^3(v_j + u_j) \bigg)  \partial_{\beta_{\sigma(1)}} \\
				&\quad - \sum_{\sigma} f_{_{\sigma(1)}} \bigg(\partial_{\beta_{\sigma(2)}} \bigg( \sum_{j=1}^3 (v_j + u_j) \bigg)\partial_{\beta_{\sigma(3)}} +\partial_{\beta_{\sigma(3)}} \bigg( \sum_{j=1}^3 (v_j + u_j) \bigg)\partial_{\beta_{\sigma(2)}} \bigg) \,.
			\end{aligned}
		\end{equation}
		We substitute the expression for $\sum_{i=1}^3 \Delta(g_i) \partial_{\beta_i}$ from Lemma \ref{lemma: properties of gs part 3} and the expression for 
		$\partial_{\nabla(h^{(\RN{1})} + h^{(\RN{2})}+ h^{(\RN{3})})}$ from Lemma~\ref{lemma: properties of h part 1} into the formula \eqref{eq: first order}.
		Then the expression \eqref{eq: first order} 
		can be rearranged as
			\begin{equation}
					\begin{aligned} \label{eq: first order rearranged}
						&- \frac12 \sum_{\sigma}  \langle \beta_{\sigma(2)}, \beta_{\sigma(3)} \rangle \widehat{u}_{\sigma(2)} \widehat{u}_{\sigma(3)} \partial_{\beta_{\sigma(1)}}
						+ \sum_{\sigma} \bigg( \sum_{j  \neq \sigma(1)} \widehat{u}_j \bigg) \big(g_{\sigma(1)}^{(\RN{2})} + g_{\sigma(1)}^{(\RN{3})} \big) \partial_{\beta_{\sigma(1)}} \\
						&\quad- \sum_{\sigma} \bigg( f_{\sigma(2)} \partial_{\beta_{\sigma(3)}} + f_{\sigma(3)} \partial_{\beta_{\sigma(2)}} \bigg)( v_{\sigma(1)} + u_{\sigma(1)} )  \partial_{\beta_{\sigma(1)}}
						- 2 \partial_{\nabla(h^{(\RN{4})})} \,,
					\end{aligned}
				\end{equation}
		which is zero by Lemma \ref{lemma: first and zeroth order} and Corollary \ref{cor: nabla of X}.
	\end{proof}

    The following lemma is needed in order to consider the zero order terms in the intertwining relation.
    \begin{lemma} \label{lemma: zero order rearranged}
    	The zero order terms satisfy $$(H  \mathcal{D} - \mathcal{D}  H_0)^{(0)} = A + B + C + D \,,$$
    	where
    	\begin{align}
    	    A &=  \sum_{i=1}^3 \bigg( \sum_{j \neq i} \widehat{u}_j \bigg) \partial_{\beta_i} \bigg( g_i^{(\RN{3})} \bigg) - \sum_{\sigma} f_{_{\sigma(1)}} \partial_{\beta_{\sigma(2)}}  \partial_{\beta_{\sigma(3)}} \bigg( \sum_{j \neq \sigma(1)} u_j  \bigg) \label{eq: A} \,, \\
    	    B &= -\frac12 \sum_{\sigma} \langle \beta_{\sigma(2)}, \beta_{\sigma(3)} \rangle \widehat{u}_{\sigma(2)} \widehat{u}_{\sigma(3)} f_{\sigma(1)}
    	    + \sum_{\sigma} \bigg( \sum_{j \neq \sigma(1)} \widehat{u}_j \bigg) \big( g_{\sigma(1)}^{(\RN{2})} + g_{\sigma(1)}^{(\RN{3})} \big) f_\sigma(1) \,, \label{eq: B} \\
    	    C &= \bigg( \sum_{j=1}^3 \widehat{u}_j \bigg) h^{(\RN{4})} - \Delta(h^{(\RN{4})}) \,, \label{eq: C} \\
    	    D &= -\sum_{i=1}^3 g_i \partial_{\beta_i} \bigg( \sum_{j  \neq i} (v_j + u_j) \bigg) \,.
    	\end{align}
    	Moreover, the term $D$ can be rearranged as $D = D_1 + D_2 \,,$ where
    	\begin{align} 
    	    &D_1 = -\sum_{i=1}^3 (g_i^{(\RN{2})} + g_i^{(\RN{3})}) \partial_{\beta_i} \bigg( \sum_{j  \neq i} (v_j + u_j) \bigg) \label{eq: D1} \,, \\
    	    &D_2 = - \sum_{\sigma} f_{_{\sigma(1)}} \bigg( f_{_{\sigma(2)}} \partial_{\beta_{\sigma(3)}} +  f_{_{\sigma(3)}} \partial_{\beta_{\sigma(2)}} \bigg) \big(v_{\sigma(1)} + u_{\sigma(1)} \big) \label{eq: D2} \,.
    	\end{align}
    \end{lemma}
    \begin{proof}
    	We have 
    	\begin{equation} \label{eq: zeroth order}
    	    \begin{aligned}
    	        (H \circ \mathcal{D})^{(0)} - (\mathcal{D} \circ H_0)^{(0)} &= 
    	        - \Delta(h) + \bigg( \sum_{j=1}^3 \widehat{u}_j \bigg) h
             	- \partial_{\beta_1}  \partial_{\beta_2} \partial_{\beta_3} \bigg( \sum_{j=1}^3 u_j \bigg) \\
    	        &\quad - \sum_{\sigma} f_{_{\sigma(1)}} \partial_{\beta_{\sigma(2)}} \partial_{\beta_{\sigma(3)}} \bigg( v_{_{\sigma(1)}} + \sum_{j=1}^3 u_j  \bigg) \\
    	        &\quad - \sum_{i=1}^3 g_i \partial_{\beta_i} \bigg( \sum_{j  \neq i} v_j +  \sum_{j=1}^3 u_j \bigg) \,.
    	    \end{aligned}
    	\end{equation}
    	By putting in the results of Lemma \ref{lemma: properties of h part 2}, the expression \eqref{eq: zeroth order} takes the required form $A + B + C + D \,.$
        By expanding $g_i = g_i^{(\RN{2})} + g_i^{(\RN{3})} + g_i^{(\RN{1})}$, we also have that
        \begin{align*}
            D = D_1 - \sum_{\sigma} f_{\sigma(2)} f_{\sigma(3)} \partial_{\beta_{\sigma(1)}} \bigg( \sum_{j  \neq \sigma(1)} (v_j + u_j) \bigg) = D_1 + D_2 
        \end{align*}
        as required.
\end{proof}
    
    In the next Lemmas \ref{lemma: A} -- \ref{lemma: D1} we rearrange the expressions for the zero order terms $A$, $B$, $C$, $D$. Namely, we rewrite these terms explicitly as functions of $\beta_j$.
    \begin{lemma} \label{lemma: A}
    	The function $A$ given by expression \eqref{eq: A} can be rearranged as follows:
    	\begin{equation} \label{eq: zeroth order sinh^4 terms}
    	\begin{aligned}
    	&\frac{A}{48 \omega^6} = 2 m(3m+1)^2 X
    	-  m (3m+1) \sum_{\sigma} \frac{\langle \beta_{\sigma(2)}, \beta_{\sigma(3)} \rangle \coth \beta_{\sigma(1)}}{\cosh^2 \beta_{\sigma(2)} \cosh^2 \beta_{\sigma(3)}}  \\
    	&\qquad + 3 m^2 (3m+1) \bigg( \sum_{j=1}^3 \frac{1}{\sinh^2 \beta_j} \bigg) X + 8 m (3m+1) \bigg( \sum_{j=1}^3 \frac{1}{\sinh^2 \beta_j} \bigg) Y \\
    	&\qquad + 24m(3m+1) Y. 
    	\end{aligned}
    	\end{equation}
    \end{lemma}
    \begin{proof}
    	Consider the term involving $g_1^{(\RN{3})}$ in $\sum_{i=1}^3 \bigg( \sum_{j \neq i} \widehat{u}_j \bigg) \partial_{\beta_i} \bigg( g_i^{(\RN{3})} \bigg)$. It gives
    	\begin{equation}\label{eq: zeroth order sinh^4 terms I}
    	\begin{aligned}
\frac{1}{48 \omega^8}    	\bigg( \sum_{j \neq 1} \widehat{u}_j \bigg) \partial_{\beta_1} \bigg( g_1^{(\RN{3})} \bigg) &= -m(3m+1)^2 \bigg( \frac{1}{\sinh^2 \beta_2} + \frac{1}{\sinh^2 \beta_3} \bigg) \frac{\coth \beta_1}{\sinh^2 \beta_1} \\
    	&\quad + m(3m+1)\bigg( \frac{1}{\cosh^2 \beta_2} + \frac{1}{\cosh^2 \beta_3} \bigg) \frac{\coth \beta_1}{\sinh^2 \beta_1} \,.
    	\end{aligned}
    	\end{equation}
    	Now consider the terms involving $u_1$ in $-\sum_{\sigma} f_{_{\sigma(1)}} \partial_{\beta_{\sigma(2)}}  \partial_{\beta_{\sigma(3)}} \bigg( \sum_{j \neq \sigma(1)} u_j  \bigg) $. They produce 
    	\begin{equation} \label{eq: zeroth order sinh^4 terms II}
    	\begin{aligned}
    	&- \frac{1}{48 \omega^8}\big(f_2 \partial_{\beta_3} \partial_{\beta_1} (u_1) + f_3 \partial_{\beta_1} \partial_{\beta_2}(u_1) \big)= \\
    	&\quad = -2m(3m+1)^2 \big(\coth \beta_2 - \coth \beta_3 \big) \frac{\coth^2 \beta_1}{\sinh^2 \beta_1} -2m(3m+1) \big(\tanh \beta_2 - \tanh \beta_3 \big) \frac{\coth^2 \beta_1}{\sinh^2 \beta_1} \\
    	&\qquad + m(3m+1) \bigg( (3m+1) \big( \coth \beta_3 - \coth \beta_2 \big) + \big( \tanh \beta_3 - \tanh \beta_2 \big) \bigg) \frac{1}{\sinh^4 \beta_1} \,.
    	\end{aligned}
    	\end{equation}
    	It follows from Lemma \ref{lemma: first order identity} (namely, equalities \eqref{eq: first order identity 3} and \eqref{eq: first order identity 4} multiplied by $\coth \beta_1$) that the sum of the right-hand side of equality \eqref{eq: zeroth order sinh^4 terms I} with the first two terms in the right-hand side of equality~\eqref{eq: zeroth order sinh^4 terms II}~equals
    	\begin{equation}  \label{eq: zeroth order sinh^4 terms III}
    	-m(3m+1)^2 \frac{\coth \beta_1}{\sinh^2 \beta_2 \sinh^2 \beta_3} - m(3m+1) \frac{\coth \beta_1}{\cosh^2 \beta_2 \cosh^2 \beta_3} \,,
    	\end{equation}
    	Adding expression \eqref{eq: zeroth order sinh^4 terms III} with analogous ones coming from the terms $g_2^{(\RN{3})}$, $u_2$, and $g_3^{(\RN{3})}$, $u_3$ in the left-hand side of equality~\eqref{eq: zeroth order sinh^4 terms} we get the 1st line of the right-hand side of equality~\eqref{eq: zeroth order sinh^4 terms} by Lemma~\ref{lemma: alt expression for X}.
    	
    	Now we rearrange the other terms in the right-hand side of equality \eqref{eq: zeroth order sinh^4 terms II}. We have that
    	\begin{equation} \label{eq: zeroth order sinh^4 terms IV}
    	\begin{aligned}
    	& m(3m+1) \bigg( (3m+1) \big( \coth \beta_3 - \coth \beta_2 \big) + \big( \tanh \beta_3 - \tanh \beta_2 \big) \bigg) \frac{1}{\sinh^4 \beta_1} \\
    	&= m(3m+1) \bigg( \frac{3m \sinh \beta_1}{\sinh \beta_2 \sinh \beta_3} + \frac{\sinh \beta_1}{\sinh \beta_2 \sinh \beta_3} - \frac{\sinh \beta_1}{\cosh \beta_2 \cosh \beta_3} \bigg) \frac{1}{\sinh^4 \beta_1} \\
    	&= \frac{3m^2(3m+1)\omega^{-2}}{\sinh^2 \beta_1}X + \frac{m(3m+1)\cosh \beta_1}{\sinh^3 \beta_1 \sinh \beta_2 \sinh \beta_3 \cosh \beta_2 \cosh \beta_3} \\
    	&=  \frac{3m^2(3m+1)\omega^{-2}}{\sinh^2 \beta_1}X + \frac{m(3m+1) (1 + \sinh^2 \beta_1)}{\sinh^3 \beta_1 \sinh \beta_2 \sinh \beta_3 \cosh \beta_1 \cosh \beta_2 \cosh \beta_3} \\
    	&= \frac{3m^2(3m+1)\omega^{-2}}{\sinh^2 \beta_1}X + \frac{8m(3m+1)\omega^{-2}}{\sinh^2 \beta_1}Y + 8m(3m+1)\omega^{-2}Y \,.
    	\end{aligned}
    	\end{equation}
    	Similarly, terms from the right-hand side of a version of equality \eqref{eq: zeroth order sinh^4 terms II} for $g_2^{(\RN{3})}$, $u_2$, and $g_3^{(\RN{3})}$, $u_3$, add up to 
    	\begin{equation} \label{eq: zeroth order sinh^4 terms V}
    	3m^2(3m+1)\omega^{-2}\bigg( \frac{1}{\sinh^2 \beta_2} + \frac{1}{\sinh^2 \beta_3}  \bigg) X + 8m(3m+1)\omega^{-2}\bigg( \frac{1}{\sinh^2 \beta_2} + \frac{1}{\sinh^2 \beta_3} \bigg)Y + 16m(3m+1)\omega^{-2}Y \,.
    	\end{equation}
    	The sum of terms in the last line of \eqref{eq: zeroth order sinh^4 terms IV} together with \eqref{eq: zeroth order sinh^4 terms V} equals the terms in the 2nd and 3rd line in the right-hand side of equality~\eqref{eq: zeroth order sinh^4 terms}.
    \end{proof} 
    
    \begin{lemma}\label{lemma: B + D2}
    	Consider the functions $B$ and $D_2$ given by \eqref{eq: B} and \eqref{eq: D2}. We have
    	\begin{align} \label{eq: B + D2}
    	   \frac{B + D_2}{64 \omega^6} &= -3m(3m+1)(3m+2) X - 12m(3m+1)Y -16(3m+1)Y \nonumber \\
    	    &\quad  + \sum_{\sigma} \frac{12m(3m+1) \langle \beta_{\sigma(2)}, \beta_{\sigma(3)} \rangle \coth \beta_{\sigma(1)}}{\sinh^2 2\beta_{\sigma(2)} \sinh^2 2 \beta_{\sigma(3)}} \,.  
    	\end{align}
    \end{lemma}
    \begin{proof}
    	As a consequence of Lemma \ref{lemma: first and zeroth order} we have
    	\begin{equation} \label{eq: B+D2}
    	\frac{B + D_2}{64 \omega^6} = - \sum_{\sigma} \frac{3m(3m+1)\omega^{-2} \langle \beta_{\sigma(2)}, \beta_{\sigma(3)} \rangle}{4\sinh^2 \beta_{\sigma(2)} \sinh^2 \beta_{\sigma(3)}}  f_{\sigma(1)} - \sum_{\sigma} \frac{2 (3m+1)\omega^{-2} \langle \beta_{\sigma(2)}, \beta_{\sigma(3)} \rangle}{\sinh^2 2\beta_{\sigma(2)} \sinh^2 2\beta_{\sigma(3)}} f_{\sigma(1)} \,.
    	\end{equation}
    	We substitute $f_j$ given by \eqref{def: fs} into the first sum in the relation \eqref{eq: B+D2}, and we substitute $$f_j = -2\omega^2(3m \coth \beta_j + 2 \coth 2\beta_j)\,,$$ $j=1,2,3,$ in the second sum in the relation \eqref{eq: B+D2}.
    	By Lemma \ref{lemma: alt expression for X}, as well as its version with $\beta_j$ replaced with $2 \beta_j$, and by Lemma~\ref{lemma: alt expression for X and Y}, we can rearrange the right-hand side of \eqref{eq: B+D2} into the required form.
    \end{proof}
    
    \begin{lemma} \label{lemma: C}
    	The function $C$ given by \eqref{eq: C} can be rearranged as follows:
    	\begin{equation} \label{eq: delta of hIV}
    	\begin{aligned}
    	    &\frac{C}{32 \omega^6} =  - 9 m^2 (3m+1) \bigg( \sum_{j=1}^3 \frac{1}{\sinh^2 \beta_j} \bigg) X + 3 m (3m+1) \bigg(2 + \sum_{j=1}^3 \frac{1}{\cosh^2 \beta_j} \bigg) X\\ 
          	&\qquad - 12 m (3m+1) \bigg( \sum_{j=1}^3 \frac{1}{\sinh^2 \beta_j} \bigg) Y  + 32 (3m+1)Y \,.
    	\end{aligned}
    	\end{equation}
    \end{lemma}
    \begin{proof}
    	By Lemma \ref{lemma: delta of X}, we have that
    	\begin{equation} \label{eq: delta of hIV rearranged}
    	\begin{aligned} 
    	    - \frac{\Delta(h^{(\RN{4})})}{32 \omega^6}  = 3m (3m+1)\bigg(2 + \sum_{j=1}^3 \frac{1}{\sinh^2 \beta_j} \bigg)X + 16 (3m+1)\bigg(2 + \sum_{j=1}^3 \frac{1}{\sinh^2 2 \beta_j} \bigg)Y  \,.
    	\end{aligned} 
    	\end{equation}
    	The product $\frac{1}{32 \omega^6} ( \sum_{j=1}^3 \widehat{u}_j ) h^{(\RN{4})}$ can be rearranged as
    	\begin{equation} \label{eq: hIV times sum of us}
    	\begin{aligned}
    	    &- \bigg( 3m \sum_{j=1}^3 \frac{1}{\sinh^2 \beta_j} + \sum_{j=1}^3 \frac{1}{\sinh^2 \beta_j} - \sum_{j=1}^3 \frac{1}{\cosh^2 \beta_j} \bigg) 3m(3m+1)X \\
    	    &- \bigg( 3m \sum_{j=1}^3 \frac{1}{\sinh^2 \beta_j} + 4 \sum_{j=1}^3 \frac{1}{\sinh^2 2\beta_j} \bigg) 4(3m+1)Y  \,.
    	\end{aligned}
    	\end{equation}
    	The equality \eqref{eq: delta of hIV} follows by multiplying \eqref{eq: hIV times sum of us} out and combining it with \eqref{eq: delta of hIV rearranged}.
    \end{proof}

    \begin{lemma} \label{lemma: D1}
    	The function $D_1$ given by \eqref{eq: D1} can be rearranged as
    	\begin{equation} \label{eq: zeroth order gs term}
	    	D_1 = 144m^2(3m+1)\omega^6\bigg( 2 + \sum_{j=1}^3 \frac{1}{\sinh^2 \beta_j} \bigg)X \,.
    	\end{equation}
    \end{lemma}
    
    \begin{proof}
    	Note that Lemma \ref{lemma: zeroth order identity 3} can be restated in the following form:
    	\[
    	   \sum_{i=1}^3 g_i^{(\RN{2})} \partial_{\beta_i} \bigg( \sum_{j  \neq i}v_j \bigg) = 0  \,.         
    	\]
    	Consider the terms in $ - \sum_{i=1}^3 g_i^{(\RN{3})} \partial_{\beta_i} \bigg( \sum_{j  \neq i} v_j \bigg)$ that involve $v_1$, that is $j = 1$. They equal
    	\begin{equation}\label{eq: zeroth order gs term I}
    	    108 m^2 (m+1)(3m+1)\omega^8 \bigg( \frac{1}{\sinh^2 \beta_3} - \frac{1}{\sinh^2 \beta_2} \bigg) \frac{\coth \alpha_1}{\sinh^2 \alpha_1} \,.
    	\end{equation}
    	Now let us look at terms in $ - \sum_{i=1}^3 g_i^{(\RN{2})} \partial_{\beta_i} \bigg( \sum_{j  \neq i} u_j \bigg)$ that involve $g_1^{(\RN{2})}$. These terms are equal to
    	\begin{equation}\label{eq: zeroth order gs term II}
    	    108 m^2 (m+1)(3m+1) \omega^8 \bigg( \frac{\coth \beta_3}{\sinh^2 \beta_3} - \frac{\coth \beta_2}{\sinh^2 \beta_2} \bigg)\frac{1}{\sinh^2 \alpha_1} \,.
    	\end{equation}
    	By Lemma \ref{lemma: zeroth order identity 1} the sum of expressions \eqref{eq: zeroth order gs term I} and \eqref{eq: zeroth order gs term II} equals 
    	\begin{equation} \label{eq: zeroth order gs term I+II}
    	    108 m^2 (m+1)(3m+1)\omega^8 \frac{\coth \beta_3 - \coth \beta_2}{\sinh^2 \beta_2 \sinh^2 \beta_3} \,.
    	\end{equation}
    	We also have that $- \sum_{i=1}^3 g_i^{(\RN{3})} \partial_{\beta_i} \bigg( \sum_{j  \neq i} u_j \bigg)$ is equal to
    	\begin{equation} \label{eq: zeroth order gs term III}
    	     36 m^2 (3m+1)^2 \omega^8 \bigg( \frac{\coth \beta_2 -\coth \beta_3}{\sinh^2 \beta_2 \sinh^2 \beta_3} - \frac{\coth \beta_1 +\coth \beta_3}{\sinh^2 \beta_1 \sinh^2 \beta_3} + \frac{\coth \beta_2 -\coth \beta_1}{\sinh^2 \beta_1 \sinh^2 \beta_2} \bigg) \,.
    	\end{equation}
    	By adding expression \eqref{eq: zeroth order gs term I+II} and the first term of expression \eqref{eq: zeroth order gs term III} we get 
    	\begin{equation}
    	    72 m^2 (3m+1)\omega^8 \frac{\coth \beta_3 - \coth \beta_2}{\sinh^2 \beta_2 \sinh^2 \beta_3} \,.
    	\end{equation}
    	By grouping similarly the remaining terms in the left-hand side of identity \eqref{eq: zeroth order gs term} and by using variants of Lemma \ref{lemma: zeroth order identity 1} obtained by rotating $\beta$'s by $\pm \frac{\pi}{3}$ (see Remark \ref{rem: variants of identities}), we get that the left-hand side of \eqref{eq: zeroth order gs term} can be rearranged as
    	\begin{equation*} 
    	\begin{aligned}
    	    & 72 m^2 (3m+1)\omega^8 \bigg( \frac{\coth \beta_3 - \coth \beta_2}{\sinh^2 \beta_2 \sinh^2 \beta_3} + \frac{\coth \beta_1 +\coth \beta_3}{\sinh^2 \beta_1 \sinh^2 \beta_3} + \frac{\coth \beta_1 -\coth \beta_2}{\sinh^2 \beta_1 \sinh^2 \beta_2} \bigg) \\
            &=72 m^2 (3m+1) \omega^6	\sum_{\sigma} \langle \beta_{\sigma(2)}, \beta_{\sigma(3)} \rangle \bigg( \frac{1}{\sinh^2 \beta_{\sigma(2)}} + \frac{1}{\sinh^2 \beta_{\sigma(3)}} \bigg) \frac{\coth\beta_{\sigma(1)}}{\sinh^2 \beta_{\sigma(1)}}  \,.
    	\end{aligned}
    	\end{equation*}
    	The result follows by Lemma \ref{lemma: zeroth order identity 2}.
    \end{proof}

	\begin{proposition}
		The zero order terms satisfy $$(H  \mathcal{D})^{(0)}  = (\mathcal{D}  H_0)^{(0)} \,.$$
	\end{proposition}
    \begin{proof}
    	By Lemmas \ref{lemma: zero order rearranged} -- \ref{lemma: D1} we have 
    	\begin{equation} \label{eq: A+B+C+D}
    		\begin{aligned}
    		    &\frac{(H \mathcal{D} - \mathcal{D} H_0)^{(0)} }{48m(3m+1)\omega^6} =
	         	2 \bigg(\sum_{j=1}^3 \frac{1}{\cosh^2 \beta_j} \bigg) X - 2X + 8 Y - \sum_{\sigma} \frac{\langle \beta_{\sigma(2)}, \beta_{\sigma(3)} \rangle \coth \beta_{\sigma(1)}}{\cosh^2 \beta_{\sigma(2)} \cosh^2 \beta_{\sigma(3)}}  \\
	    		&\qquad\qquad\qquad\qquad 
+ \sum_{\sigma} \frac{16 \langle \beta_{\sigma(2)}, \beta_{\sigma(3)} \rangle \coth \beta_{\sigma(1)}}{\sinh^2 2\beta_{\sigma(2)} \sinh^2 2 \beta_{\sigma(3)} } \,.
    		\end{aligned}
    	\end{equation}
    	Let us replace $\sinh^{-2} 2\beta_{\sigma(2)} \sinh^{-2} 2\beta_{\sigma(3)}$ in the last sum in \eqref{eq: A+B+C+D} with
    	\[
    	    \frac{1}{16} \bigg( \frac{1}{\sinh^2 \beta_{\sigma(2)}} - \frac{1}{\cosh^2 \beta_{\sigma(2)}} \bigg) \bigg( \frac{1}{\sinh^2 \beta_{\sigma(3)}} - \frac{1}{\cosh^2 \beta_{\sigma(3)}} \bigg) \,.
    	\]
    	By using Lemma \ref{lemma: alt expression for X} the right-hand side of \eqref{eq: A+B+C+D} can be rewritten as
    	$E + F$, where
    	\begin{equation}
    		E = -4X + 8Y \,
    	\end{equation}
    	and
    	\begin{equation}
	    	\begin{aligned}
		    	F = 2\bigg(\sum_{j=1}^3 \frac{1}{\cosh^2 \beta_j} \bigg)X 
-\sum_{\sigma\in S_3}\frac{\langle \beta_{\sigma(2)}, \beta_{\sigma(3)}\rangle \coth \beta_{\sigma(1)}}{\sinh^2 \beta_{\sigma(2)} \cosh^2 \beta_{\sigma(3)}},
	    	\end{aligned}
    	\end{equation}
where the summation is over the symmetric group.
    	Let us collect terms with $\cosh^{-2} \beta_1$ in  $F$. We have
    	\begin{align}
    	    &\bigg(2X + \frac{\omega^2\coth \beta_2}{\sinh^2 \beta_3} - \frac{\omega^2\coth \beta_3 }{\sinh^2 \beta_2} \bigg)\frac{1}{\cosh^2 \beta_1}  \nonumber \\
    	    & = \bigg(\frac{\sinh(\beta_2 - \beta_1) + \sinh \beta_1 \cosh \beta_2}{\sinh \beta_3} + \frac{\sinh(\beta_1 + \beta_3) - \sinh \beta_1 \cosh \beta_3}{\sinh\beta_2 } \bigg)\frac{X}{\cosh^2 \beta_1}  \nonumber \\
    	    & =  \bigg(\frac{\sinh\beta_2}{\sinh \beta_3} + \frac{\sinh\beta_3}{\sinh \beta_2}\bigg)\frac{X}{\cosh \beta_1}  = \omega^2\bigg( \frac{1}{\sinh^2 \beta_2} + \frac{1}{\sinh^2 \beta_3} \bigg) \frac{\tanh \beta_1}{\sinh^2 \beta_1} \,. \label{eq: cosh^-2 term}
    	\end{align} 
    	Note that by multiplying the relation \eqref{eq: first order identity 4} in Lemma \ref{lemma: first order identity} through by $\tanh \beta_1$ we rearrange \eqref{eq: cosh^-2 term} as
    	\begin{equation} \label{eq: cosh^-2 term rearranged}
    	    \bigg( \frac{1}{\sinh^2 \beta_2} + \frac{1}{\sinh^2 \beta_3} \bigg) \frac{\tanh \beta_1}{\sinh^2 \beta_1} = \frac{\tanh \beta_1}{\sinh^2 \beta_2 \sinh^2 \beta_3} - \frac{2 (\coth \beta_2 - \coth \beta_3) }{\sinh^2 \beta_1} = \frac{\tanh \beta_1}{\sinh^2 \beta_2 \sinh^2 \beta_3} + 2 \omega^{-2}X \,.
    	\end{equation}
    	Similarly, we collect and rearrange terms in $F$ with $\cosh^{-2} \beta_2$ and $\cosh^{-2} \beta_3$. Then by using variants of the identity \eqref{eq: cosh^-2 term rearranged} obtained by rotating $\beta$'s (see Remark \ref{rem: variants of identities}) we get 
    	\begin{align*}
    	    F 
    	   = \sum_{\sigma} \frac{\langle \beta_{\sigma(2)}, \beta_{\sigma(3)} \rangle \tanh \beta_{\sigma(1)}}{\sinh^2 \beta_{\sigma(2)} \sinh^2 \beta_{3}} + 6X = 4X - 8Y 
    	\end{align*}
    	by Lemma \ref{lemma: alt expression for X and Y}. Hence $E + F = 0$ as required.
    \end{proof}

\section{Rational limit}
\label{rational_limit}

    In the rational limit $\omega \to 0$ the operator $H_0$ takes the form
    \begin{equation}\label{eq: H_0 rational}
	    H_0^r = -\Delta + \sum_{i=1}^{3} \bigg( \frac{m (m + 1) \langle \tilde \alpha_i, \tilde  \alpha_i \rangle }{\langle\tilde \alpha_i, x \rangle^2}  
	    + \frac{3m (3 m + 1) \langle \tilde \beta_i, \tilde \beta_i \rangle}{\langle \tilde  \beta_i, x \rangle^2} \bigg), 
    \end{equation}
where vectors $\tilde \alpha_i$, $\tilde \beta_i$ can be taken as the original vectors $\alpha_i, \beta_i$ with any fixed non-zero value of $\omega$. 
    The Hamiltonian $H$ in the rational limit becomes
    \begin{equation}\label{eq: H rational}
    	H^r = -\Delta + \sum_{i=1}^{3} \bigg( \frac{m (m + 1) \langle \tilde \alpha_i, \tilde \alpha_i \rangle }{\langle \tilde \alpha_i, x \rangle^2}  
    	+ \frac{(3 m + 1) (3m + 2) \langle \tilde \beta_i, \tilde \beta_i \rangle}{\langle \tilde \beta_i, x \rangle^2} \bigg).
    \end{equation}
    And the explicit form of the intertwining operator $\mathcal D$ in the rational limit is the operator ${\mathcal D}^r = \lim_{\omega\to 0} \omega^{-3} \mathcal D$ which takes the form 
	\begin{equation} \label{eq: D rational}
		\begin{aligned}
			&\mathcal{D}^r = \partial_{\tilde \beta_{1}} \partial_{\tilde \beta_{2}} \partial_{\tilde \beta_{3}} - \sum_{\sigma} \frac{(3m+1)\langle \tilde \beta_{\sigma(1)}, \tilde \beta_{\sigma(1)} \rangle }{\langle \tilde \beta_{\sigma(1)}, x \rangle} \partial_{\tilde \beta_{\sigma(2)}} \partial_{\tilde \beta_{\sigma(3)}} + \sum_{\sigma}  \frac{ (3m+1)^2\langle \tilde  \beta_{\sigma(1)}, \tilde \beta_{\sigma(1)} \rangle^2 }{{\langle \tilde \beta_{\sigma(2)}, x \rangle} {\langle \tilde \beta_{\sigma(3)}, x \rangle}} \partial_{\tilde \beta_{\sigma(1)}} \\ 
			&\qquad- \sum_{i=1}^3 \bigg(  \frac{m(m+1)\prod_{k \neq i} \langle \tilde \alpha_{i}, \tilde \beta_k \rangle}{\langle \tilde \alpha_{i}, x \rangle^2} +  \frac{3m(3m+1) \prod_{k \neq i} \langle \tilde \beta_{i}, \tilde \beta_k \rangle}{\langle \tilde \beta_{i}, x \rangle^2} \bigg)  \partial_{\tilde \beta_{i}} \\ 
			&\qquad + \sum_{i=1}^3  \frac{9m(m+1)(3m+1) \prod_{k = 1}^3 \langle \tilde \beta_i, \tilde \beta_k \rangle}{\langle \tilde  \beta_i, x \rangle^3} 
			+ \sum_{i=1}^3   \frac{m(m+1)(3m+1) \langle \tilde \beta_i, \tilde \beta_i \rangle \prod_{k \neq i} \langle \tilde  \alpha_i, \tilde \beta_k \rangle}{\langle \tilde \beta_i, x \rangle \langle \tilde \alpha_i, x \rangle^2}  \\
			&\qquad -  \frac{3 (3 m + 1) (6 m^2 + 6 m + 1)\prod_{i = 1}^3 \langle \tilde \beta_i, \tilde \beta_i \rangle}{2 \langle \tilde \beta_{1}, x \rangle \langle \tilde  \beta_{2}, x \rangle \langle \tilde \beta_{3}, x \rangle}.
		\end{aligned}
	\end{equation}

	\begin{theorem}\label{rat_lim_th}
		The operators defined by formulas \eqref{eq: H_0 rational}, \eqref{eq: H rational} and \eqref{eq: D rational} satisfy the intertwining relation
		\begin{equation}
			H^r \mathcal{D}^r = \mathcal{D}^r H_0^r.
		\end{equation}
	\end{theorem}

Similarly to the trigonometric case we derive quantum integrals in the factorised form, following \cite{Ch}:
$$
[H^r, {\mathcal D^r {D^r}^*}]=0, \quad [H_0^r, {{\mathcal D}^r}^* \mathcal D^r]=0.
$$

\begin{remark}
The operator $H^r$ is the ordinary $G_2$ Calogero--Moser operator with multiplicities $m$, $3m+1$. Therefore the operator $D^r$ is the rational version of the corresponding Opdam's shift operator in a suitable gauge for the $G_2$-orbit containing  $\tilde\beta_i$ \cite{Opd}.  Hence the operator $D^r$ can also be constructed via the product of the corresponding (rational) Dunkl operators $\nabla_{\tilde \beta_1} \nabla_{\tilde \beta_2}\nabla_{\tilde \beta_3}$ as it was demonstrated by Heckman for any root system in \cite{HeckRat}. 
\end{remark}

\section{Concluding remarks}
\label{concluderem}

We established integrability of the CMS system associated with the collection of vectors $AG_2$ and an arbitrary value of the parameter $m$. This configuration of vectors is interesting as it is an example of a slightly weakened notion of a root system. Indeed, the configuration is invariant under the Weyl group $G_2$ and the root vectors belong to the invariant lattice but the crystallgraphic condition between the root vectors is no longer satisfied. This makes it harder to study the corresponding CMS  system as, for instance, we could not define (trigonometric) Dunkl operators with good properties for the model $AG_2$. Nonetheless integrability property appears to be present.

There are a number of further questions about this system. Firstly, it is natural to consider elliptic version and investigate its integrability. 
Secondly, it would be interesting to clarify whether the classical analogue of the sytem is integrable. In the case of root system $G_2$ Lax pairs for the corresponding CMS model were constructed in \cite{HP}, \cite{BS} (see also \cite{FM}), which may be a starting point for approaching classical $AG_2$ CMS system. Another approach could be to investigate classical version of the quantum integral ${\mathcal D \mathcal D}^*$. On the other hand let us consider the operator $\hbar^2 H$ and take the limit $\hbar\to 0, m\to \infty$ such that $\hbar m \to const$. It is easy to see that the resulting classical Hamiltonian is the ordinary $G_2$ Hamiltonian. This suggests that the classical analogue of $H$ where potential is the same as in the quantum case may be non-integrable. 

Thirdly, it would be interesting to investigate bispectrality of the considered Hamiltonian $H$. More specifically, existence of the intertwining operator  $\mathcal D$ implies that for integer $m$ the Hamiltonian $H$ has Baker--Akhiezer eigenfunction 
$$
\psi(k,x) =  \mathcal D \phi(k,x), \quad  H \psi (k,x) = (k_1^2+k_2^2) \psi(k,x),
$$
where $\phi(k,x)$ is the Baker--Akhiezer function for $G_2$ CMS system \cite{CV, CSV},  and $k=(k_1, k_2)$ is the spectral parameter. Bispectral dual Hamiltonian, if exists, would be an operator of Ruijsenaars--Macdonald type acting in $k$-variables of $\psi(k,x)$ so that $\psi(k,x)$ is its eigenfunction. In the root system case and for type $A$ deformed CMS system such type of bispectrality is established in \cite{Chbisp} (see also \cite{Fbisp} for other examples).

We hope to return to some of these questions soon.

\end{document}